\documentclass[journal]{IEEEtran}
\usepackage[utf8]{inputenc}
\usepackage{cite}
\usepackage{amsmath,amssymb,amsfonts,amsthm}
\usepackage{graphicx}
\usepackage{algorithm}
\usepackage{algorithmic}
\usepackage {multirow}
\usepackage{bm}
\usepackage{textcomp}
\usepackage{url}
\usepackage{mathtools}
\usepackage{xcolor}
\usepackage{stfloats}
\usepackage{placeins}
\usepackage{booktabs} % For professional tables
\usepackage{tabularx}
\usepackage{balance}
\newtheorem{theorem}{Theorem}

\newtheorem{proposition}{Proposition}
\newtheorem{definition}{Definition}

\newtheorem{remark}{Remark}
\begin{document}

\title{Generalized Linear Graph Representation: A Compact Operator Space for Graph Signal Processing and Graph Neural Networks}

\author{Feiyue~Zhao, Zhichao~Zhang,~\IEEEmembership{Member,~IEEE}, and Yangfan~He
\vspace{-0.25em} 
\thanks{This work was supported in part by the Open Foundation of Hubei Key Laboratory of Applied Mathematics (Hubei University) under Grant HBAM202404; in part by the Foundation of Key Laboratory of System Control and Information Processing, Ministry of Education under Grant Scip20240121; and in part by the Startup Foundation for Introducing Talent of Nanjing Institute of Technology under Grant YKJ202214. \emph{(Corresponding author: Zhichao~Zhang.)}}
\thanks{Feiyue~Zhao is with the School of Mathematics and Statistics, Nanjing University of Information Science and Technology, Nanjing 210044, China (e-mail: 202511150010@nuist.edu.cn).}
\thanks{Zhichao~Zhang is with the School of Mathematics and Statistics, Nanjing University of Information Science and Technology, Nanjing 210044, China, with the Hubei Key Laboratory of Applied Mathematics, Hubei University, Wuhan 430062, China, and also with the Key Laboratory of System Control and Information Processing, Ministry of Education, Shanghai Jiao Tong University, Shanghai 200240, China (e-mail: zzc910731@163.com).}
\thanks{Yangfan~He is with the School of Communication and Artificial Intelligence, School of Integrated Circuits, Nanjing Institute of Technology, Nanjing 211167, China, and also with the Jiangsu Province Engineering Research Center of IntelliSense Technology and System, Nanjing 211167, China (e-mail: Yangfan.He@njit.edu.cn).}}

\markboth{}{}
\maketitle
% \vspace{-5pt}
\begin{abstract}
Graph Signal Processing (GSP) and Graph Neural Networks (GNNs) rely fundamentally on the matrix representation of the underlying graph topology. This representation defines key operators such as the graph Fourier transform, spectral filtering, and convolution. Existing parameterized operator families interpolate only partial subsets of classical graph matrices, while broader formulations become non-compact when representing transition-type operators, limiting both theoretical analysis and stable learning. To address this issue, we propose the Generalized Linear Graph Representation (GLGR), denoted by $\mathbf{Q}_{\alpha,l}$, as a compact two-parameter operator family defined on a bounded linear domain. GLGR unifies major classical operators together with transition-type operators without requiring asymptotic parameters. Theoretically, we show that $\mathbf{Q}_{\alpha,l}$ admits a variational decomposition balancing local smoothness and global degree-weighted energy, derive spectral perturbation bounds, and establish graph-aware sufficient conditions for positive semi-definiteness. Building on this formulation, we develop Adaptive GLGR Convolution (AG-Conv), which makes the propagation operator itself learnable within end-to-end GNNs. Experiments on graph classification and node classification benchmarks show that GLGR improves both fixed-operator representation search and adaptive graph learning across multiple backbones. 

\end{abstract}


\section{Introduction}

\IEEEPARstart{G}{raph}-structured data are now central to scientific computing and machine intelligence, spanning social systems, biological interaction networks, citation graphs, and molecular structures. In this context, Graph Signal Processing (GSP) provides the mathematical foundation for extending classical signal operations to irregular domains \cite{ortega2018graph,shuman2013emerging,Sandryhaila2013,isufi2024graph,dong2020graph}. A core concept in GSP is the \textit{graph shift operator} (GSO), because the selected matrix representation determines the graph Fourier basis, spectral filtering behavior, and ultimately the inductive bias of downstream models.

\begin{figure}[t]
    \centering
    \includegraphics[width=0.95\linewidth]{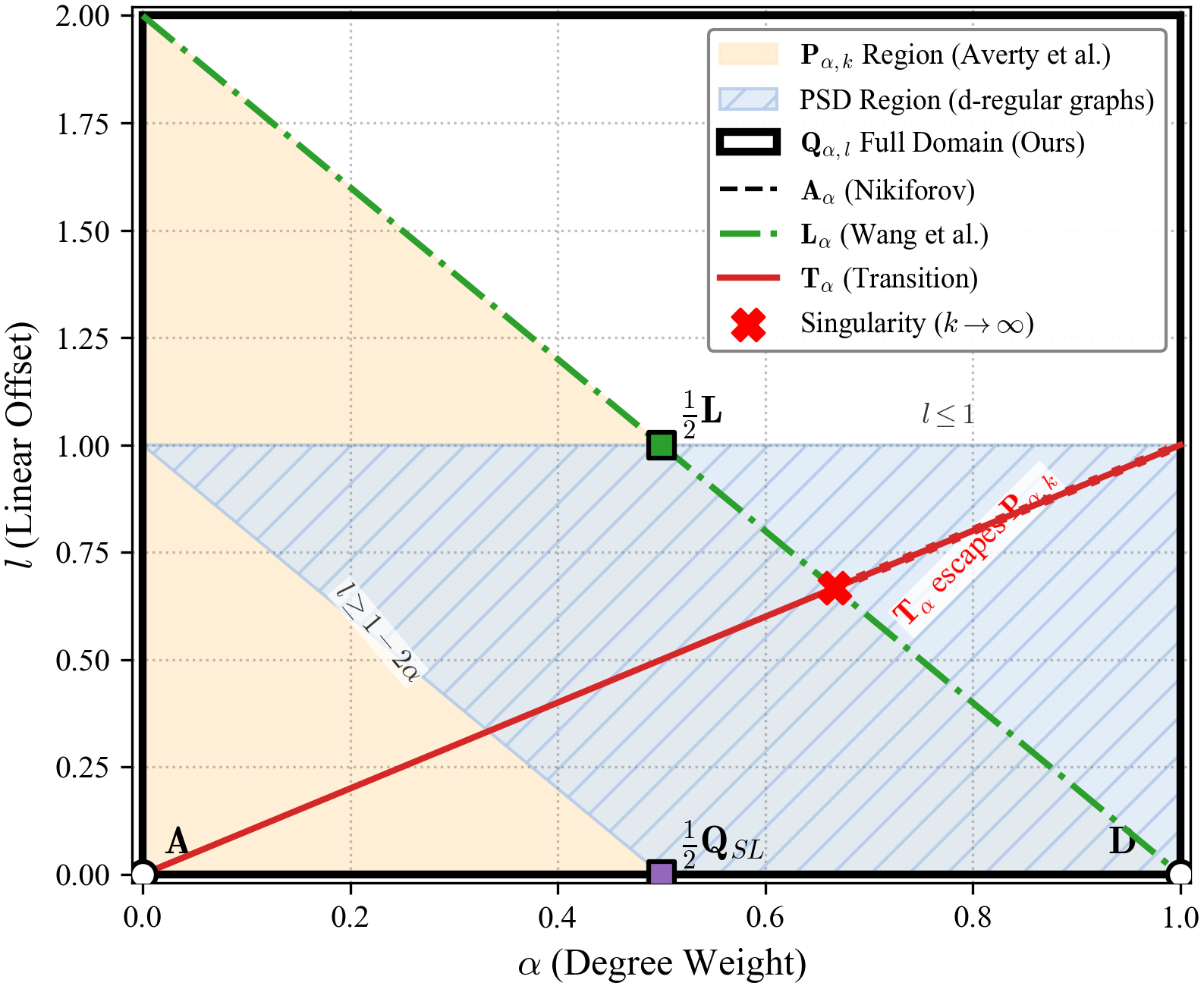}
    \caption{The parameter space of the Generalized Linear Graph Representation ($\mathbf{Q}_{\alpha, l}$). The shaded blue area indicates an illustrative PSD-guaranteed region obtained from the simplified $d$-regular-graph analysis in Section~IV; the exact PSD domain is graph dependent. The diagram highlights the unification of standard matrices ($\mathbf{A}, \mathbf{D}, \mathbf{L}, \mathbf{Q}_{SL}$) and parametric families ($\mathbf{A}_{\alpha}, \mathbf{L}_{\alpha}, \mathbf{T}_{\alpha}$) within the linear $(\alpha, l)$ domain.}
    \label{fig:parameter_space}
\end{figure}

This choice is equally fundamental in Graph Neural Networks (GNNs), whose development has evolved from geometric deep learning and early spectral convolution layers to expressive attention-based and benchmarked modern pipelines \cite{bronstein2017geometric,kipf2017semi,velivckovic2017graph,wu2020comprehensive,dwivedi2023benchmarking}. Although many modern GNN layers are presented in message-passing form, their propagation rules can be interpreted as polynomial filters of an underlying shift operator. Therefore, the matrix representation is not only a theoretical object in GSP but also a practical design variable controlling smoothing strength, high-frequency retention, and optimization dynamics in deep graph models. In particular, homophilous graphs typically benefit from smooth low-pass propagation, whereas heterophilous graphs often require preserving or amplifying higher-frequency components, a distinction closely connected to over-smoothing and low-pass bias analyses \cite{li2018deeper,oono2020graph,nt2019revisiting}.

Existing representations, however, remain fragmented. The adjacency-centered view from algebraic signal processing \cite{Sandryhaila2013} and the Laplacian-centered view from spectral geometry \cite{shuman2013emerging} capture complementary but incompatible priors. Parametric families such as $\mathbf{A}_{\alpha}$ \cite{Nikiforov2017} and $\mathbf{L}_{\alpha}$ \cite{Wang2020} partially bridge specific operator pairs, yet they do not provide a unified, compact, and optimization-friendly space covering both schools. More recently, $\mathbf{P}_{\alpha,k}$ \cite{Averty2024} aims to enlarge this space, but its non-linear parameter coupling introduces a structural singularity: representing transition-type operators requires $k \rightarrow \infty$. Such non-compactness is especially problematic for learnable architectures, where bounded and well-conditioned parameter domains are crucial for stable gradient-based training.

From a broader perspective, operator design is deeply coupled with graph construction, regularization, and spectral prior modeling. Representative directions include graph kernels and regularization, manifold/Laplacian embedding, trend filtering, smooth-signal graph learning, heat-diffusion graph identification, and vertex-domain graph operations \cite{smola2003kernels,belkin2003laplacian,wang2016trend,kalofolias2016learn,thanou2017learning,kartal2021graph}. Complementary studies on focused spectral analysis, sparse matrix spectral ordering, and random geometric graph regimes further highlight that topology-dependent spectra can vary substantially across graphs \cite{vandeville2017when,barnard1995spectral,penrose2003random}. Recent theoretical and application-driven GSP advances also reinforce the need for robust and adaptive graph operators \cite{song2024graph,morgenstern2024theoretical,buchnik2024gsp,castro2024gegenbauer}.

\begin{figure*}[t]
    \centering
    \includegraphics[width=\linewidth]{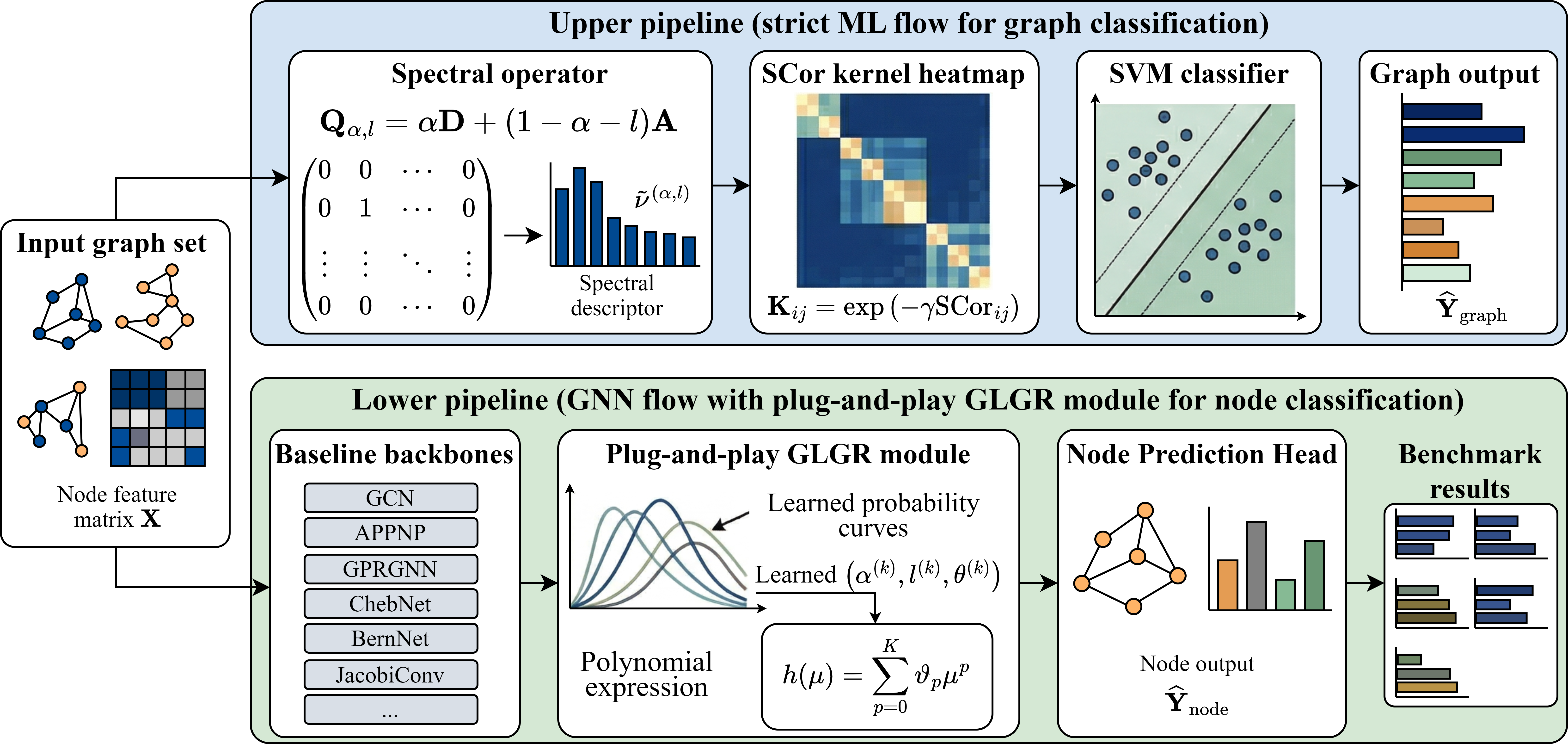}
    \caption{Two complementary experimental tracks in this study. \textbf{Upper branch (graph classification):} a fixed-parameter representation-space pipeline based on GLGR spectral descriptors, SCor-based kernel construction, and SVM classification. \textbf{Lower branch (node classification):} a learnable-parameter pipeline where Adaptive GLGR Convolution (AG-Conv) is plugged into baseline backbones for end-to-end node-level prediction.}
    \label{fig:glgr_overview}
\end{figure*}

Recent long-range graph architectures further reinforce this point. Graph transformers and graph-adapted state space models enlarge the receptive field and improve global interaction, but their performance remains strongly influenced by structural encoding and propagation bias \cite{dwivedi2021generalization,ying2021do,rampasek2022recipe,wu2022nodeformer,wu2023sgformer,chen2023nagphormer,ma2024polyformer,wang2024graphmamba,behrouz2024graphmamba}. In other words, architectural expressiveness does not remove the need for a well-behaved graph operator; it makes operator choice even more consequential. This shifts the central question from designing yet another architecture around a fixed matrix to constructing a compact and learnable operator space itself.

Motivated by this need for a unified and learnable operator space, we propose the Generalized Linear Graph Representation (GLGR), denoted by $\mathbf{Q}_{\alpha,l}$. The core design is to linearly decouple diagonal and off-diagonal components through an independent offset parameter, so that different structural priors can be controlled without inducing singular behavior. As a result, GLGR forms a bounded and well conditioned operator manifold that can recover classical matrices while remaining suitable for gradient-based optimization. Fig.~\ref{fig:parameter_space} provides a geometric illustration of this compact operator family in the $(\alpha,l)$ plane.

GLGR is intended to serve both spectral analysis and learnable graph propagation. On the GSP side, it yields a unified operator family with clear variational meaning and analyzable stability properties. On the GNN side, the same parameterization can be learned as a propagation operator, allowing the model to adapt its frequency preference across homophilous and heterophilous regimes. Fig.~\ref{fig:glgr_overview} summarizes these two experimental tracks, including an upper graph-classification pipeline based on GLGR spectral descriptors and a lower node-classification pipeline with end-to-end adaptive GLGR integration. The main contributions of this paper are summarized as follows:
\begin{enumerate}
    \item \textbf{Compact Operator Unification in a Bounded Domain:} We propose the GLGR family $\mathbf{Q}_{\alpha,l}$ and show that it places major classical graph operators in a bounded linear domain, while representing transition-type operators without singular limits. This resolves the non-compactness issue of prior non-linear parameterizations.
    \item \textbf{Variational and Stability Theory for GLGR:} We derive a generalized graph energy interpretation of $\mathbf{Q}_{\alpha,l}$ and establish spectral perturbation bounds together with graph-aware sufficient PSD conditions, providing a rigorous explanation of why GLGR is both expressive and numerically stable.
    \item \textbf{Learnable GLGR for Graph Neural Networks:} We develop Adaptive GLGR Convolution (AG-Conv), where $(\alpha,l)$ are learned jointly with network weights. This enables task-dependent frequency adaptation and consistently improves multiple backbones on both homophilous and heterophilous node classification benchmarks.
    \item \textbf{Two-Level Empirical Validation:} Through exhaustive SCor+SVM grid search and end-to-end GNN training, we show that GLGR offers both representation-space gains (better fixed operators) and optimization-space gains (learnable and transferable improvements).
\end{enumerate}

The remainder of this paper is organized as follows. Section II reviews standard operators and formalizes the limitations of existing unified matrices. Section III introduces GLGR and proves its unification properties. Section IV develops the spectral theory of GLGR, including energy decomposition, perturbation stability, and positive semi-definiteness conditions. Section V presents the adaptive GLGR convolution for GNNs. Section VI reports graph and node classification experiments, and Section VII concludes the paper.

\section{Related Work and Problem Formulation}

\subsection{From Spectral Operators to Graph Neural Propagation}
Let $\mathcal{G}=(\mathcal{V},\mathcal{E})$ be an undirected graph with $N=|\mathcal{V}|$ nodes, adjacency matrix $\mathbf{A}$, and degree matrix $\mathbf{D}$ with $\mathbf{D}_{ii}=\sum_j\mathbf{A}_{ij}$. Classical graph operators include the combinatorial Laplacian $\mathbf{L}=\mathbf{D}-\mathbf{A}$ and the signless Laplacian $\mathbf{Q}_{SL}=\mathbf{D}+\mathbf{A}$.

In GSP, a GSO $\mathbf{S}\in\mathbb{R}^{N\times N}$ defines the spectral coordinate system: if $\mathbf{S}=\mathbf{U}\mathbf{\Lambda}\mathbf{U}^\top$, then graph Fourier analysis is induced by $\mathbf{U}$, and spectral filters are $g(\mathbf{S})=\mathbf{U}g(\mathbf{\Lambda})\mathbf{U}^\top$. This perspective is directly aligned with GNNs, where many message-passing layers are equivalent to polynomial filters $\sum_{k=0}^{K}c_k\mathbf{S}^k$. Hence, the choice of $\mathbf{S}$ jointly controls spectral bias, propagation behavior, and optimization characteristics.

\subsection{Existing Unified Families and Their Limitations}
To interpolate among canonical operators, prior work introduced parameterized families. Nikiforov's family $\mathbf{A}_{\alpha}=\alpha\mathbf{D}+(1-\alpha)\mathbf{A}$ \cite{Nikiforov2017} unifies adjacency- and degree-oriented views, while Wang's family $\mathbf{L}_{\alpha}=\alpha\mathbf{D}+(\alpha-1)\mathbf{A}$ \cite{Wang2020} bridges Laplacian and degree operators. These constructions are useful but remain separated: neither provides a single compact framework simultaneously covering both adjacency- and Laplacian-centered representations. A later attempt is
\begin{equation} \label{eq:P_alpha_k}
    \mathbf{P}_{\alpha,k}=\alpha\mathbf{D}+(2k-1)(\alpha-1)\mathbf{A},\quad \alpha,k\in[0,1],
\end{equation}
which enlarges the interpolation space but introduces non-linear parameter coupling. Consider the transition-type operator
\begin{equation}
    \mathbf{T}_{\alpha}=\alpha\mathbf{D}+(1-2\alpha)\mathbf{A}.
\end{equation}
Matching coefficients between $\mathbf{P}_{\alpha,k}$ and $\mathbf{T}_{\alpha}$ yields
\begin{equation}
    k=\frac{1}{2}+\frac{1-2\alpha}{2(\alpha-1)}.
\end{equation}
When $\alpha\to1$, $k\to\infty$. Therefore, transition-type operators lie on an asymptotic boundary rather than in the interior of a compact parameter domain. This non-compactness is problematic in learning scenarios: bounded parameterizations cannot reach the target region, while unconstrained optimization becomes numerically unstable.

\subsection{Problem Formulation and Design Criteria}
The above limitations indicate that a practically useful unified family must satisfy more than algebraic interpolation. We formulate the target as follows: construct a two-parameter operator family $\mathbf{M}(\theta)$ that simultaneously satisfies
\begin{enumerate}
    \item \textbf{Compact coverage:} major classical and transition-type operators are represented within a bounded domain;
    \item \textbf{Optimization well-posedness:} parameterization avoids singular mappings and supports stable gradient flow;
    \item \textbf{GNN compatibility:} the operator can be inserted directly into learnable propagation layers for both homophilous and heterophilous settings.
\end{enumerate}

In other words, we seek a unified representation where spectral interpretability and end-to-end learnability are achieved in the same parameter space. Section III introduces GLGR to meet these criteria.

\section{Generalized Linear Graph Representation}

\subsection{Definition and Geometry}

\begin{definition}[GLGR Operator]
For any undirected graph $\mathcal{G}$, the Generalized Linear Graph Representation operator $\mathbf{Q}_{\alpha, l}$ is defined over the compact domain $(\alpha, l) \in [0, 1] \times [0, 2]$ as:
    \begin{equation} \label{eq:Q_def}
        \mathbf{Q}_{\alpha, l} := \alpha \mathbf{D} + (1 - \alpha - l) \mathbf{A}.
    \end{equation}
\end{definition}
Geometrically, $\mathbf{Q}_{\alpha, l}$ represents a point in the affine subspace spanned by $\mathbf{D}$ and $\mathbf{A}$. The parameter $\alpha$ controls the weight of the self-loops (degree information), while the term $(1 - \alpha - l)$ linearly modulates the edge weights. Crucially, the parameter $l$ acts as a linear offset, allowing the sign and magnitude of the adjacency component to vary smoothly without inducing singularities.

\subsection{Unification of Spectral Spaces}

\begin{theorem}[Topological Completeness]
    The parameter space $\Omega_{\text{GLGR}} = \{(\alpha, l) \mid 0 \le \alpha \le 1, 0 \le l \le 2\}$ provides a compact linear domain that covers the classical operator trajectories induced by $\mathbf{A}_{\alpha}$ and $\mathbf{L}_{\alpha}$, and represents the full transition trajectory $\mathbf{T}_{\alpha}$ without asymptotic parameters. In particular, $\mathbf{T}_{\alpha}$ is realized by the linear path $l=\alpha$ inside $\Omega_{\text{GLGR}}$.
\end{theorem}

\begin{proof}
The proof follows by constructive identification of the parameters $(\alpha, l)$ for each target matrix.
First, the Adjacency matrix $\mathbf{A}$ is recovered at the origin $(\alpha=0, l=0)$. The Degree matrix $\mathbf{D}$ corresponds to the line $l = 1-\alpha$ where the adjacency term vanishes, or simply $\alpha=1, l=0$.
Second, the Laplacian $\mathbf{L} = \mathbf{D} - \mathbf{A}$ is obtained at $(\alpha,l)=(1,1)$. Similarly, the Signless Laplacian is obtained at $(\alpha,l)=(\tfrac12,0)$ up to positive scaling, and lies on the line $l=0$ with positive adjacency coefficient.
Third, we address the parametric families. Nikiforov's $\mathbf{A}_{\alpha}$ corresponds trivially to the axis $l=0$. Wang's $\mathbf{L}_{\alpha}$ corresponds to the linear constraint $l = 2(1-\alpha)$, which maps to a valid segment within the domain since $0 \le 2(1-\alpha) \le 2$ for $\alpha \in [0, 1]$.
Finally, the transition matrix $\mathbf{T}_{\alpha} = \alpha \mathbf{D} + (1-2\alpha)\mathbf{A}$ is recovered by setting $l = \alpha$. Since $\alpha \in [0, 1]$, this trajectory lies entirely within the bounded rectangle of the parameter space. This differs from $\mathbf{P}_{\alpha, k}$, where recovering $\mathbf{T}_{\alpha}$ required infinite parameters. Thus, GLGR provides a bounded representation of these operator trajectories without asymptotic parameters.
\end{proof}

\subsection{Gradient Stability and Learnability}
While the $\mathbf{P}_{\alpha,k}$ framework theoretically unifies classical operators, its non-linear parameter coupling creates severe optimization bottlenecks in deep learning. Specifically, the gradient with respect to $\alpha$ is $\frac{\partial \mathbf{P}}{\partial \alpha} = \mathbf{D} + (2k-1)\mathbf{A}$. As the network adapts towards transition-like operators ($k \to \infty$), the gradient linearly diverges, leading to gradient explosion. Bounding $k$ with a squashing map does not remove the obstruction, the optimization suffers from zero-gradient saturation at the boundaries, rendering crucial spectral regions unreachable. Thus, the issue is not the squashing function itself, but forcing a non-compact operator family into a bounded parameterization.

By contrast, GLGR is already compact in $(\alpha,l)$. At the operator level, its gradients are strictly linear and constant matrices: $\frac{\partial \mathbf{Q}}{\partial l} = -\mathbf{A}$ and $\frac{\partial \mathbf{Q}}{\partial \alpha} = \mathbf{D} - \mathbf{A} = \mathbf{L}$. This topology-independent, bounded gradient flow avoids the singular behavior induced by non-compact parameterization and substantially improves optimization conditioning, making GLGR more amenable to stable end-to-end optimization via standard backpropagation.

\section{Spectral Analysis and Properties}
	\subsection{Physical Interpretation via Energy Decomposition}
	In GSP, the Laplacian quadratic form $\mathbf{x}^\top \mathbf{L} \mathbf{x}$ measures the “smoothness" or Dirichlet energy of a signal. However, this measure ignores the signal's global energy magnitude. We prove that $\mathbf{Q}_{\alpha, l}$ provides a physically meaningful trade-off between these two fundamental quantities.
	
	\begin{theorem}[Energy Decomposition] \label{thm:energy}
		For any graph signal $\mathbf{x} \in \mathbb{R}^N$, the quadratic form associated with $\mathbf{Q}_{\alpha, l}$ decomposes into a linear combination of the Dirichlet energy (local variation) and the Total Vertex energy (global retention):
        \begin{align}
        E_{\alpha, l}(\mathbf{x}) &= \mathbf{x}^\top \mathbf{Q}_{\alpha, l} \mathbf{x} \nonumber \\
        &= C_{\text{smooth}} \sum_{(i,j) \in \mathcal{E}} (x_i - x_j)^2 + C_{\text{global}} \sum_{i \in \mathcal{V}} d_i x_i^2,
        \end{align}
		where the coefficients are given by:
		\begin{equation}
			C_{\text{smooth}} = \alpha + l - 1, \quad C_{\text{global}} = 1 - l.
		\end{equation}
	\end{theorem}
    
	\begin{proof}
		We start with the algebraic definition of the quadratic form:
		\begin{equation}
			E_{\alpha, l}(\mathbf{x}) = \alpha \mathbf{x}^\top \mathbf{D} \mathbf{x} + (1-\alpha-l) \mathbf{x}^\top \mathbf{A} \mathbf{x}.
		\end{equation}
		Expanding the terms using the properties $\mathbf{x}^\top \mathbf{D} \mathbf{x} = \sum_i d_i x_i^2$ and $\mathbf{x}^\top \mathbf{A} \mathbf{x} = \sum_{(i,j) \in \mathcal{E}} 2 x_i x_j$:
		\begin{equation} \label{eq:expansion}
			E_{\alpha, l}(\mathbf{x}) = \alpha \sum_{i \in \mathcal{V}} d_i x_i^2 + (1-\alpha-l) \sum_{(i,j) \in \mathcal{E}} 2 x_i x_j.
		\end{equation}
		Using the edge identity
		\begin{equation}
			(x_i - x_j)^2 = x_i^2 + x_j^2 - 2x_i x_j.
		\end{equation}
		and summing over all edges $(i,j) \in \mathcal{E}$, we obtain
		\begin{equation}
			\sum_{(i,j) \in \mathcal{E}} 2 x_i x_j = \sum_{i \in \mathcal{V}} d_i x_i^2 - \sum_{(i,j) \in \mathcal{E}} (x_i - x_j)^2.
		\end{equation}
		Substituting this into Eq. (\ref{eq:expansion}) yields
        \begin{equation}
        \begin{aligned}
            &E_{\alpha, l}(\mathbf{x}) = \alpha \sum\nolimits_i d_i x_i^2 \\
            &\quad + (1-\alpha-l) \left[ \sum\nolimits_i d_i x_i^2 - \sum\nolimits_{(i,j)\in\mathcal{E}} (x_i - x_j)^2 \right] \\
            &= [\alpha + (1-\alpha-l)] \sum\nolimits_i d_i x_i^2 \\
            &\quad - (1-\alpha-l) \sum\nolimits_{(i,j)\in\mathcal{E}} (x_i - x_j)^2 \\
            &= (1-l) \sum_{i \in \mathcal{V}} d_i x_i^2  + (l + \alpha - 1) \sum_{(i,j) \in \mathcal{E}} (x_i - x_j)^2.
        \end{aligned}
        \end{equation}
        
		Therefore the claimed coefficients follow directly.
	\end{proof}

	\subsection{Spectral Stability Analysis}
	A critical requirement for robust GSP is spectral stability: small perturbations in the graph topology should not yield catastrophic changes in the representation matrix. We prove that GLGR offers superior stability compared to the $\mathbf{P}_{\alpha, k}$ family.
	\begin{proposition}
    Let $\lambda_k(\mathbf{Q}_{\alpha,l})$ denote the $k$-th eigenvalue of $\mathbf{Q}_{\alpha,l}$ sorted in nondecreasing order. For fixed $\alpha$, the mapping $l \mapsto \mathbf{Q}_{\alpha,l}$ is affine, and each ordered eigenvalue is Lipschitz continuous in $l$ with constant at most $\|\mathbf{A}\|_2$.
\end{proposition}
	    
	    \begin{proof}
	    From Eq. (\ref{eq:Q_def}),
	    \begin{equation}
	        \mathbf{Q}_{\alpha,l_2}-\mathbf{Q}_{\alpha,l_1}=-(l_2-l_1)\mathbf{A}.
	    \end{equation}
	    Weyl's inequality gives
	    \begin{equation}
	        |\lambda_k(\mathbf{Q}_{\alpha,l_2})-\lambda_k(\mathbf{Q}_{\alpha,l_1})| \le \|\mathbf{Q}_{\alpha,l_2}-\mathbf{Q}_{\alpha,l_1}\|_2 = |l_2-l_1|\,\|\mathbf{A}\|_2.
	    \end{equation}
	    Hence each eigenvalue is Lipschitz continuous in $l$.
	    \end{proof}

	\begin{theorem}[Spectral Perturbation Bound] \label{thm:stability}
		Let $\mathcal{G}$ be a graph with adjacency matrix $\mathbf{A}$. Consider a perturbed graph $\tilde{\mathcal{G}}$ with $\tilde{\mathbf{A}} = \mathbf{A} + \mathbf{E}$, where $\mathbf{E}$ represents the structural noise. Let $\lambda_k(\cdot)$ denote the $k$-th eigenvalue sorted in nondecreasing order. The spectral perturbation of the GLGR matrix is bounded by:
		\begin{equation}
			| \lambda_k(\tilde{\mathbf{Q}}_{\alpha, l}) - \lambda_k(\mathbf{Q}_{\alpha, l}) | \le \mathcal{K}_{\mathrm{GLGR}} \| \mathbf{E} \|_2,
		\end{equation}
		where the sensitivity coefficient $\mathcal{K}_{\mathrm{GLGR}} = \alpha\sqrt{N} + |1 - \alpha - l|$.
		In contrast, the sensitivity of the $\mathbf{P}_{\alpha, k}$ family diverges as $k \to \infty$.
	\end{theorem}

		\begin{proof}
			Because both $\mathcal{G}$ and $\tilde{\mathcal{G}}$ are undirected, the matrices $\mathbf{A}$, $\tilde{\mathbf{A}}$, and $\mathbf{E}=\tilde{\mathbf{A}}-\mathbf{A}$ are symmetric. Weyl's inequality therefore gives
			\begin{equation}
				|\lambda_k(\tilde{\mathbf{Q}}_{\alpha,l})-\lambda_k(\mathbf{Q}_{\alpha,l})| \le \|\Delta \mathbf{Q}\|_2,
			\end{equation}
			where $\Delta \mathbf{Q}:=\tilde{\mathbf{Q}}_{\alpha,l}-\mathbf{Q}_{\alpha,l}$.

			Since $\tilde{\mathbf{D}}-\mathbf{D}=\mathrm{diag}(\mathbf{E}\mathbf{1})$, we have
			\begin{equation}
				\Delta \mathbf{Q}=\alpha\,\mathrm{diag}(\mathbf{E}\mathbf{1})+(1-\alpha-l)\mathbf{E}.
			\end{equation}
			Hence, by the triangle inequality,
			\begin{equation}
				\|\Delta \mathbf{Q}\|_2 \le \alpha\,\|\mathrm{diag}(\mathbf{E}\mathbf{1})\|_2 + |1-\alpha-l|\,\|\mathbf{E}\|_2.
			\end{equation}

			For any vector $\mathbf{v}$, the spectral norm of the diagonal matrix $\mathrm{diag}(\mathbf{v})$ is
			\begin{equation}
				\|\mathrm{diag}(\mathbf{v})\|_2 = \|\mathbf{v}\|_\infty.
			\end{equation}
			Applying this with $\mathbf{v}=\mathbf{E}\mathbf{1}$ yields
            \begin{equation}
			\begin{aligned}
				&\|\mathrm{diag}(\mathbf{E}\mathbf{1})\|_2
				= \|\mathbf{E}\mathbf{1}\|_\infty 
				\le \|\mathbf{E}\mathbf{1}\|_2 
				\le \|\mathbf{E}\|_2\,\|\mathbf{1}\|_2 = \sqrt{N}\,\|\mathbf{E}\|_2.
			\end{aligned}
            \end{equation}
			Substituting this bound gives
			\begin{equation}
				\|\Delta \mathbf{Q}\|_2 \le \big(\alpha\sqrt{N}+|1-\alpha-l|\big)\|\mathbf{E}\|_2.
			\end{equation}
			Combining with Weyl's inequality proves
			\begin{equation}
				|\lambda_k(\tilde{\mathbf{Q}}_{\alpha,l})-\lambda_k(\mathbf{Q}_{\alpha,l})| \le \big(\alpha\sqrt{N}+|1-\alpha-l|\big)\|\mathbf{E}\|_2.
			\end{equation}

			For comparison, the adjacency coefficient in the $\mathbf{P}_{\alpha,k}$ family is $(2k-1)(\alpha-1)$. Along the transition trajectory one requires $k\to\infty$, so
			\begin{equation}
				\lim_{k\to\infty}|(2k-1)(\alpha-1)|=\infty,
			\end{equation}
			which explains why its perturbation sensitivity cannot remain uniformly bounded in the same manner.
		\end{proof}

\begin{remark}[Tightness of the degree-perturbation term]
The bound in Theorem~\ref{thm:stability} is intentionally graph agnostic and therefore conservative. Before introducing the uniform estimate $\|\mathbf{E}\mathbf{1}\|_{\infty}\le \sqrt{N}\,\|\mathbf{E}\|_2$, one already has the sharper operator bound
\begin{equation}
    \|\Delta \mathbf{Q}\|_2 \le \alpha\,\|\mathbf{E}\mathbf{1}\|_{\infty} + |1-\alpha-l|\,\|\mathbf{E}\|_2.
\end{equation}
Here $\|\mathbf{E}\mathbf{1}\|_{\infty}$ is exactly the maximum node-degree perturbation induced by the structural noise. In sparse graphs with localized edge insertions or deletions, this quantity is often much smaller than $\sqrt{N}\,\|\mathbf{E}\|_2$, so the practical perturbation level can be substantially below the worst-case bound stated in Theorem~\ref{thm:stability}.
\end{remark}

\subsection{Positive Semi-Definiteness}
In many GSP applications, particularly those involving graph kernels and smoothness regularization, the representation matrix is strictly required to be PSD. Unlike the Laplacian $\mathbf{L}$, the general GLGR matrix $\mathbf{Q}_{\alpha, l}$ is not automatically PSD due to the indefinite nature of the adjacency component. Here, we derive a conservative, graph-aware sufficient condition for PSD based on spectral bounds, and then use the $d$-regular case as an illustrative special example for geometric visualization.

\begin{theorem}[Sufficient PSD Condition]
Consider a graph $\mathcal{G}$ with minimum degree $d_{\min}$. Let the spectrum of its Adjacency matrix $\mathbf{A}$ be bounded by $[\lambda_{\min}(\mathbf{A}), \lambda_{\max}(\mathbf{A})]$. A sufficient condition for $\mathbf{Q}_{\alpha, l}$ to be positive semi-definite is:
\begin{equation} \label{eq:psd_condition}
    \alpha d_{\min} + \min \left\{ (1 - \alpha - l) \lambda_{\min}(\mathbf{A}), (1 - \alpha - l) \lambda_{\max}(\mathbf{A}) \right\} \ge 0.
\end{equation}
For $d$-regular graphs, as an illustrative special case rather than a universal PSD boundary, this condition simplifies to the region defined by the intersection of $2\alpha + l \ge 1$ and $l \le 1$.
\end{theorem}

\begin{proof}
The PSD property requires that the quadratic form $q(\mathbf{x}) = \mathbf{x}^\top \mathbf{Q}_{\alpha, l} \mathbf{x}$ be non-negative for all non-zero vectors $\mathbf{x} \in \mathbb{R}^N$. Substituting the definition of $\mathbf{Q}_{\alpha, l}$:
\begin{equation}
    q(\mathbf{x}) = \alpha \mathbf{x}^\top \mathbf{D} \mathbf{x} + (1 - \alpha - l) \mathbf{x}^\top \mathbf{A} \mathbf{x}.
\end{equation}
Let $\beta = 1 - \alpha - l$. We utilize the Rayleigh quotient properties. The term $\mathbf{x}^\top \mathbf{D} \mathbf{x}$ is bounded below by $d_{\min} \|\mathbf{x}\|^2$. The term $\mathbf{x}^\top \mathbf{A} \mathbf{x}$ lies in the interval $[\lambda_{\min}(\mathbf{A}) \|\mathbf{x}\|^2, \lambda_{\max}(\mathbf{A}) \|\mathbf{x}\|^2]$.
We analyze the condition $q(\mathbf{x}) \ge 0$ based on the sign of $\beta$:
\begin{enumerate}
    \item \textbf{Case $\beta \ge 0$ (Adjacency-dominant regime):}
    The minimum value of the quadratic form is bounded by:
    \begin{equation}
        \inf_{\mathbf{x} \neq \mathbf{0}} \frac{q(\mathbf{x})}{\|\mathbf{x}\|^2} \ge \alpha d_{\min} + \beta \lambda_{\min}(\mathbf{A}).
    \end{equation}
    Thus, a sufficient condition is $\alpha d_{\min} + (1 - \alpha - l) \lambda_{\min}(\mathbf{A}) \ge 0$.
    \item \textbf{Case $\beta < 0$ (Laplacian-dominant regime):}
    Here, the term $\beta \mathbf{x}^\top \mathbf{A} \mathbf{x}$ is minimized when $\mathbf{x}^\top \mathbf{A} \mathbf{x}$ is maximal (since $\beta$ is negative). Thus:
    \begin{equation}
        \inf_{\mathbf{x} \neq \mathbf{0}} \frac{q(\mathbf{x})}{\|\mathbf{x}\|^2} \ge \alpha d_{\min} + \beta \lambda_{\max}(\mathbf{A}).
    \end{equation}
    The condition becomes $\alpha d_{\min} + (1 - \alpha - l) \lambda_{\max}(\mathbf{A}) \ge 0$.
\end{enumerate}
Combining these cases yields Eq. (\ref{eq:psd_condition}) as a sufficient PSD condition.
For a $d$-regular graph, $d_{\min}=d$, $\lambda_{\max}(\mathbf{A})=d$, and typically $\lambda_{\min}(\mathbf{A}) \ge -d$.
In Case 1 ($\beta \ge 0 \Rightarrow l \le 1-\alpha$): $\alpha d + (1-\alpha-l)(-d) \ge 0 \Rightarrow \alpha - (1-\alpha-l) \ge 0 \Rightarrow l \ge 1 - 2\alpha$.
In Case 2 ($\beta < 0 \Rightarrow l > 1-\alpha$): $\alpha d + (1-\alpha-l)(d) \ge 0 \Rightarrow \alpha + 1 - \alpha - l \ge 0 \Rightarrow l \le 1$.
The intersection of these inequalities defines the illustrative PSD-guaranteed region depicted in Fig. \ref{fig:parameter_space}.
\end{proof}

Note that the exact PSD boundaries depend on the specific graph topology. The shaded blue region in Fig. \ref{fig:parameter_space} visualizes only the simplified $d$-regular-graph case as an illustrative example, rather than the full PSD domain for arbitrary graphs.

Beyond the PSD domain, the figure also illustrates the operator coverage of our framework. As illustrated in Fig. \ref{fig:parameter_space}, the orange shaded region represents the accessible domain of the prior $\mathbf{P}_{\alpha, k}$ framework. Notably, the trajectory of the transition matrix $\mathbf{T}_{\alpha}$ (red solid line) intersects the boundary of $\mathbf{P}_{\alpha, k}$ at $\alpha=\tfrac23$, producing a singularity point ($k \to \infty$, marked by the red cross). For $\alpha > \tfrac23$, $\mathbf{T}_{\alpha}$ lies outside the bounded $\mathbf{P}_{\alpha, k}$ domain. By comparison, our proposed GLGR $\mathbf{Q}_{\alpha, l}$ spans the full $[0, 1] \times [0, 2]$ parameter box (black bounding box), which contains the classical trajectories discussed above together with the transition path, without requiring singular parameter values.

This theorem delineates a specific region in the $(\alpha, l)$ plane where the matrix is valid for kernel operations, providing a theoretical guideline for hyperparameter search.

\section{Adaptive GLGR Graph Neural Networks}
\label{sec:gnn}

Standard Graph Convolutional Networks (GCNs) effectively perform low-pass filtering on graph signals, smoothing node features based on local homophily \cite{kipf2017semi}. However, this rigid smoothing mechanism often leads to the \textit{over-smoothing} problem and fails on \textit{heterophilic graphs}, where adjacent nodes exhibit dissimilar features (high-frequency components) \cite{li2018deeper,nt2019revisiting,zhu2020beyond}. To address these limitations, we propose the \textbf{Adaptive GLGR-GNN}, a spectral architecture that leverages the tunable spectral properties of $\mathbf{Q}_{\alpha, l}$ to learn optimal frequency response functions, ranging from low-pass to high-pass and band-pass filtering.

\subsection{Unified Operator View of Graph Neural Propagation}
Let $\mathbf{X}\in\mathbb{R}^{N\times F}$ denote node features. A broad class of graph neural layers can be written as
\begin{equation}
    \mathbf{H}^{(k+1)} = \Phi\!\left(\sum_{p=0}^{K} c_p\,\mathbf{S}^p\mathbf{H}^{(k)}\mathbf{W}_p^{(k)}\right),
\end{equation}
where $\mathbf{S}$ is a propagation operator and $\Phi$ is a nonlinearity. Different backbones mainly differ in the choice of $\mathbf{S}$ and polynomial coefficients. This abstraction exposes a common insertion point for GLGR: replacing fixed operators by a learnable, bounded operator family.

\subsection{Adaptive GLGR Convolution (AG-Conv)}
For layer $k$, we define
\begin{equation}
    \mathbf{Q}^{(k)} \triangleq \mathbf{Q}_{\alpha^{(k)},l^{(k)}} = \alpha^{(k)}\mathbf{D} + \big(1-\alpha^{(k)}-l^{(k)}\big)\mathbf{A},
\end{equation}
and use it as the propagation basis in a polynomial filter:
\begin{equation} \label{eq:poly_conv}
    \mathbf{H}^{(k+1)} = \Phi\!\left( \sum_{p=0}^K (\mathbf{Q}^{(k)})^p \mathbf{H}^{(k)} \mathbf{\Theta}_p^{(k)} \right).
\end{equation}
Here $\mathbf{\Theta}_p^{(k)} \in \mathbb{R}^{F_{in}\times F_{out}}$ and $(\alpha^{(k)},l^{(k)})$ are trainable structural parameters (shared by all channels in layer $k$). To keep them in-domain, we use a constrained parameterization in practice, e.g.,
\begin{equation}
    \alpha^{(k)} = \sigma(a^{(k)}),\quad l^{(k)} = 2\,\sigma(b^{(k)}),
\end{equation}
where $a^{(k)},b^{(k)}\in\mathbb{R}$ are unconstrained parameters and $\sigma(\cdot)$ is the Sigmoid function.

In the node-classification experiments, AG-Conv uses the learned GLGR operator $\mathbf{Q}^{(k)}$ in Eq.~(\ref{eq:poly_conv}) as its propagation basis. Relative to fixed-operator layers, it learns the structural parameters $(\alpha^{(k)},l^{(k)})$ jointly with the polynomial coefficients $\{\mathbf{\Theta}_p^{(k)}\}_{p=0}^{K}$, so both the operator and its spectral response can adapt to the data. This flexibility helps the model handle different homophily regimes without committing to a fixed low-pass prior. The next proposition shows that, once the layer parameters are fixed, the resulting propagation map still obeys a linear perturbation bound inherited from the underlying GLGR operator.

\begin{proposition}[Layer-Wise Perturbation Bound for AG-Conv]\label{prop:agconv_stability}
Consider the $k$-th AG-Conv layer with fixed structural parameters $(\alpha^{(k)},l^{(k)})$, polynomial coefficients $\{\mathbf{\Theta}_p^{(k)}\}_{p=0}^{K}$, and an $L_{\Phi}$-Lipschitz activation function $\Phi$. Let $\mathbf{H}\in\mathbb{R}^{N\times F_{\mathrm{in}}}$ be the input feature matrix. Suppose the original graph is perturbed by structural noise $\mathbf{E} = \tilde{\mathbf{A}} - \mathbf{A}$. Let $\mathcal{T}_{\mathbf{Q}^{(k)}}$ and $\mathcal{T}_{\tilde{\mathbf{Q}}^{(k)}}$ denote the pre-activation polynomial propagation maps of the original and perturbed graphs, respectively.

Define the maximum operator spectral norm $R_k := \max\!\left\{\|\mathbf{Q}^{(k)}\|_2,\,\|\tilde{\mathbf{Q}}^{(k)}\|_2\right\}$. Then, the layer-wise output perturbation is bounded by:
\begin{equation}\label{eq:agconv_graph_stability}
    \begin{aligned}
        &\big\|\Phi(\mathcal{T}_{\tilde{\mathbf{Q}}^{(k)}}(\mathbf{H}))-\Phi(\mathcal{T}_{\mathbf{Q}^{(k)}}(\mathbf{H}))\big\|_{F} \\
        &\le
        C_k \, L_{\Phi} \, \|\mathbf{H}\|_{F} \, \mathcal{K}_{\mathrm{GLGR}}^{(k)} \, \|\mathbf{E}\|_2,
	\end{aligned}
\end{equation}
where $\mathcal{K}_{\mathrm{GLGR}}^{(k)} := \alpha^{(k)}\sqrt{N}+\big|1-\alpha^{(k)}-l^{(k)}\big|$ is the structural sensitivity coefficient, and $C_k := \sum_{p=1}^{K} p\,R_k^{p-1}\,\|\mathbf{\Theta}_p^{(k)}\|_2$ captures the polynomial amplification factor.
\end{proposition}

\begin{proof}
The proof proceeds by bounding the perturbation from the operator level up to the full layer output. Let $\Delta\mathbf{Q}^{(k)} := \tilde{\mathbf{Q}}^{(k)}-\mathbf{Q}^{(k)}$.

\textit{Step 1: Polynomial operator difference.}
Since the $p=0$ term is identical for both operators, only orders $p\ge 1$ contribute to the difference. Using the telescoping sum identity, we have
\begin{equation}
    (\tilde{\mathbf{Q}}^{(k)})^p-(\mathbf{Q}^{(k)})^p
    = \sum_{t=0}^{p-1}(\tilde{\mathbf{Q}}^{(k)})^{p-1-t}\,\Delta\mathbf{Q}^{(k)}\,(\mathbf{Q}^{(k)})^t.
\end{equation}
Applying the submultiplicativity of the spectral norm and the definition of $R_k$ yields:
\begin{equation}\label{eq:poly_diff_bound}
    \begin{aligned}
        &\big\|(\tilde{\mathbf{Q}}^{(k)})^p-(\mathbf{Q}^{(k)})^p\big\|_2 \\
        &\le \sum_{t=0}^{p-1}\big\|(\tilde{\mathbf{Q}}^{(k)})^{p-1-t}\big\|_2\,\big\|\Delta\mathbf{Q}^{(k)}\big\|_2\,\big\|(\mathbf{Q}^{(k)})^t\big\|_2 \\
        &\le \sum_{t=0}^{p-1} R_k^{p-1-t}\,\big\|\Delta\mathbf{Q}^{(k)}\big\|_2\,R_k^t
        = p\,R_k^{p-1}\,\big\|\Delta\mathbf{Q}^{(k)}\big\|_2.
    \end{aligned}
\end{equation}

\textit{Step 2: Pre-activation linear propagation bound.}
The pre-activation map is defined as $\mathcal{T}_{\mathbf{Q}^{(k)}}(\mathbf{H}) := \sum_{p=0}^{K} (\mathbf{Q}^{(k)})^p \mathbf{H}\mathbf{\Theta}_p^{(k)}$. By the mixed-norm inequality $\|\mathbf{A}\mathbf{B}\mathbf{C}\|_{F} \le \|\mathbf{A}\|_2\,\|\mathbf{B}\|_{F}\,\|\mathbf{C}\|_2$ and Eq.~\eqref{eq:poly_diff_bound}, we obtain:
\begin{equation}\label{eq:agconv_linear_stability}
\begin{aligned}
    &\big\|\mathcal{T}_{\tilde{\mathbf{Q}}^{(k)}}(\mathbf{H})-\mathcal{T}_{\mathbf{Q}^{(k)}}(\mathbf{H})\big\|_{F} \\
    &\le \sum_{p=1}^{K}\big\|(\tilde{\mathbf{Q}}^{(k)})^p-(\mathbf{Q}^{(k)})^p\big\|_2\,\|\mathbf{H}\|_{F}\,\|\mathbf{\Theta}_p^{(k)}\|_2 \\
    &\le \|\mathbf{H}\|_{F}\left(\sum_{p=1}^{K} p\,R_k^{p-1}\,\|\mathbf{\Theta}_p^{(k)}\|_2\right)\big\|\Delta\mathbf{Q}^{(k)}\big\|_2 \\
    &= C_k \, \|\mathbf{H}\|_{F} \, \big\|\Delta\mathbf{Q}^{(k)}\big\|_2.
\end{aligned}
\end{equation}

\textit{Step 3: Full layer output bound.}
Given that the activation function $\Phi$ is $L_{\Phi}$-Lipschitz with respect to the Frobenius norm, the output difference is bounded by:
\begin{equation}\label{eq:agconv_full_stability}
\begin{aligned}
    &\big\|\Phi(\mathcal{T}_{\tilde{\mathbf{Q}}^{(k)}}(\mathbf{H}))-\Phi(\mathcal{T}_{\mathbf{Q}^{(k)}}(\mathbf{H}))\big\|_{F} \\
    &\le L_{\Phi} \, \big\|\mathcal{T}_{\tilde{\mathbf{Q}}^{(k)}}(\mathbf{H})-\mathcal{T}_{\mathbf{Q}^{(k)}}(\mathbf{H})\big\|_{F}.
\end{aligned}
\end{equation}
Finally, substituting Eq.~\eqref{eq:agconv_linear_stability} into Eq.~\eqref{eq:agconv_full_stability}, and applying the operator-norm estimate $\big\|\Delta\mathbf{Q}^{(k)}\big\|_2 \le \mathcal{K}_{\mathrm{GLGR}}^{(k)}\,\|\mathbf{E}\|_2$ (established in Theorem~\ref{thm:stability}), we arrive at the desired bound in Eq.~\eqref{eq:agconv_graph_stability}.
\end{proof}

\subsection{Transfer to Diverse Baselines: Common Principle and Model-Specific Instantiation}
The common principle is straightforward: keep each backbone's optimization pipeline, depth, residual/teleport mechanism, and non-propagation modules unchanged, and only replace its core propagation operator by the learned GLGR operator $\mathbf{Q}_{\alpha,l}$ (or inject it into the corresponding polynomial/diffusion kernel).

Concretely, the reported node-classification models follow three reusable templates:
\begin{itemize}
    \item \textbf{First-order message passing} (e.g., GCN/SGC): replace single-step propagation by $\mathbf{Q}_{\alpha,l}$.
    \item \textbf{Polynomial spectral models} (e.g., ChebNet, BernNet, JacobiConv, GPRGNN): keep polynomial order and basis design, but replace the fixed Laplacian/adjacency kernel with powers of $\mathbf{Q}_{\alpha,l}$.
    \item \textbf{Diffusion/propagation models} (e.g., APPNP): replace the fixed propagation matrix in personalized diffusion recursion by $\mathbf{Q}_{\alpha,l}$ while preserving teleport/residual design.
\end{itemize}
This yields a plug-and-play upgrade path: architecture-specific heads remain unchanged, while structural propagation becomes adaptive and data-driven. For backbones whose original derivations rely on a specific normalized operator, polynomial approximation interval, or diffusion matrix---such as ChebNet, BernNet, JacobiConv, GPRGNN, and APPNP---we keep the surrounding architectural template and training pipeline unchanged and replace only the propagation kernel by $\mathbf{Q}_{\alpha,l}$. For cross-backbone spectral-response visualization, we report responses in a shared normalized GLGR coordinate $\bar{\mu}\in[-1,1]$ obtained by affine rescaling of the injected operator spectrum. Table~\ref{tab:backbone_instantiation} summarizes the exact insertion rule used for each backbone in the node-classification experiments.

\begin{table*}[t]
\centering
\caption{Backbone-specific operator-injection instantiation of AG-Conv in node classification. All reported results directly use the GLGR operator $\mathbf{Q}_{\alpha,l}$. The table specifies where $\mathbf{Q}_{\alpha,l}$ is inserted and which architecture-specific mechanisms are kept unchanged.}
\label{tab:backbone_instantiation}
{\footnotesize
\setlength{\tabcolsep}{4pt}
\renewcommand{\arraystretch}{1.08}
\begin{tabularx}{\textwidth}{l|X|X|X}
\toprule
Backbone & GLGR insertion & Components kept unchanged & Parameter sharing \\
\midrule
SGC & Replace the fixed $K$-hop propagation by repeated application of $\mathbf{Q}_{\alpha,l}$. & Hop count, linear classifier, optimizer, and early-stopping rule. & One $(\alpha,l)$ pair for the propagation block, shared across hops and channels. \\
GCN & Replace the first-order propagation operator in each graph layer by $\mathbf{Q}_{\alpha,l}$. & Layer depth, activations, dropout, hidden width, optimizer, and training schedule. & One pair per graph-convolution layer, shared across channels. \\
GPRGNN & Use $\mathbf{Q}_{\alpha,l}$ as the propagation basis in the original GPR polynomial. & Polynomial order, coefficient parameterization, skip/readout design, optimizer, and training schedule. & One pair per propagation block, shared across polynomial orders and channels. \\
APPNP & Replace the fixed diffusion matrix in personalized propagation by $\mathbf{Q}_{\alpha,l}$. & Teleport probability, MLP predictor, iteration count, residual/teleport path, optimizer, and training schedule. & One pair per propagation block, shared across iterations and channels. \\
ChebNet & Build the Chebyshev recursion on $\mathbf{Q}_{\alpha,l}$. & Chebyshev order, coefficient learning, classifier head, optimizer, and training schedule. & One pair per layer, shared across polynomial orders and channels. \\
BernNet & Build the Bernstein polynomial kernel on $\mathbf{Q}_{\alpha,l}$. & Bernstein basis design, order, classifier head, optimizer, and training schedule. & One pair per layer, shared across polynomial orders and channels. \\
JacobiConv & Build the Jacobi polynomial kernel on $\mathbf{Q}_{\alpha,l}$. & Jacobi basis hyperparameters, order, readout head, optimizer, and training schedule. & One pair per layer, shared across polynomial orders and channels. \\
SGFormer & Replace only the local graph-propagation branch by $\mathbf{Q}_{\alpha,l}$; keep the transformer branch untouched. & Attention branch, branch fusion, residual links, optimizer, and training schedule. & One pair per graph branch layer, shared across channels. \\
\bottomrule
\end{tabularx}}
\end{table*}

Accordingly, we do not claim that these GLGR-enhanced variants preserve the original normalization semantics, diffusion interpretation, or approximation-theoretic guarantees attached to their parent operators. Rather, they should be read as empirically matched replacements that test whether a single bounded operator family can serve as a common propagation basis across otherwise unchanged architectures.

\begin{figure}[t]
    \centering
    \includegraphics[width=\linewidth]{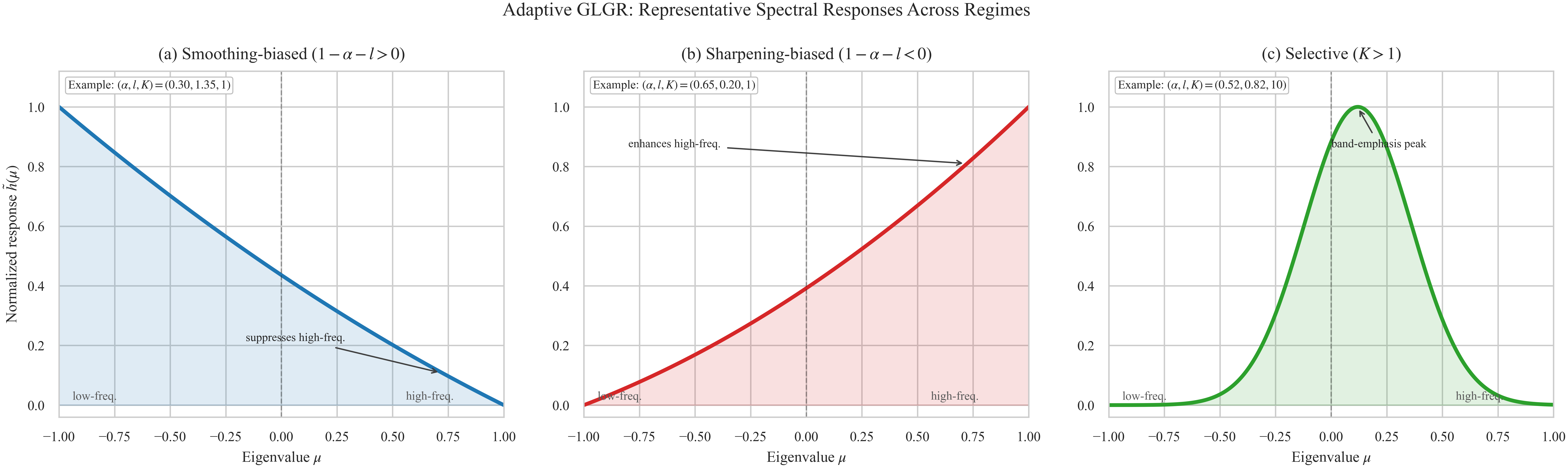}
    \caption{\textbf{Representative spectral responses of Adaptive GLGR.} Illustrative response curves $h(\mu)=\sum_{p=0}^{K}\vartheta_p\mu^p$ are shown on the eigenvalue axis of $\mathbf{Q}_{\alpha,l}$ for three regimes: \textbf{(a)} smoothing-biased ($1-\alpha-l>0$), \textbf{(b)} sharpening-biased ($1-\alpha-l<0$), and \textbf{(c)} selective ($K>1$). Curves are illustrative parameter settings; dataset-specific profiles are learned during training.}
    \label{fig:filter_response}
\end{figure}

\subsection{Frequency Response and Homophily Adaptation}
Let $\mu_i$ denote an eigenvalue of $\mathbf{Q}_{\alpha,l}$. For notational clarity, we use $\{\vartheta_p\}$ (and, layer-wise, $\{\vartheta_p^{(k)}\}$) for scalar coefficients in the one-dimensional spectral response, reserving $\{\mathbf{\Theta}_p^{(k)}\}$ for the matrix-valued trainable weights in Eq.~(\ref{eq:poly_conv}). The scalar spectral response of AG-Conv is
\begin{equation}
    h(\mu_i)=\sum_{p=0}^{K}\vartheta_p\mu_i^p.
\end{equation}
Hence, the final filtering behavior is jointly determined by the operator spectrum $\{\mu_i\}$ and scalar response coefficients $\{\vartheta_p\}$. Importantly, the eigenvalue range is \emph{operator-dependent} rather than fixed: for example, normalized Laplacian eigenvalues lie in $[0,2]$, while normalized adjacency-type operators are often in $[-1,1]$; for the GLGR operator $\mathbf{Q}_{\alpha,l}=\alpha\mathbf{D}+(1-\alpha-l)\mathbf{A}$, the spectrum depends on graph topology and on $(\alpha,l)$, and is therefore not universally constrained to either interval. In addition, AG-Conv does not produce only one response curve: each layer has its own learned $(\alpha^{(k)},l^{(k)})$ and polynomial coefficients $\{\vartheta_p^{(k)}\}$, yielding a layer-specific response $h^{(k)}(\mu)$ (and potentially different responses across runs/datasets). In particular, the term $(1-\alpha-l)$ controls the sign and strength of off-diagonal coupling in $\mathbf{Q}_{\alpha,l}$, while $\alpha$ controls degree-weighted diagonal contribution. Combined with Theorem~\ref{thm:energy} (where $C_{\text{smooth}}=\alpha+l-1=-(1-\alpha-l)$), this yields the following practically useful interpretation:
\begin{itemize}
    \item \textbf{Smoothing-biased regime} ($1-\alpha-l>0$): neighborhood averaging is encouraged, typically favoring low-frequency consistency and often beneficial on homophilous graphs.
    \item \textbf{Sharpening-biased regime} ($1-\alpha-l<0$): contrast across connected nodes is relatively strengthened, which can help preserve boundary/high-frequency components and is often advantageous under heterophily.
    \item \textbf{Selective regime} ($K>1$): higher-order polynomials provide non-monotone responses (e.g., band emphasis), enabling task-dependent frequency selectivity.
\end{itemize}
These regimes indicate spectral bias, while the exact profile is determined by the learned pair $\big(\mathbf{Q}_{\alpha,l},\{\vartheta_p\}\big)$. Fig.~\ref{fig:filter_response} gives representative patterns.

\section{Experimental Validation}

To comprehensively validate GLGR, we design two complementary experimental tracks: (i) \emph{fixed-parameter graph classification} using SCor+SVM with exhaustive grid search in the $(\alpha,l)$ plane, and (ii) \emph{learnable-parameter node classification} by integrating AG-Conv into multiple GNN backbones. This design is intentional: the first track evaluates the \emph{representation capacity} of GLGR itself, while the second track evaluates whether this capacity can be \emph{automatically exploited} by end-to-end learning.

\begin{figure*}[!t]
    \centering
    \includegraphics[width=\linewidth]{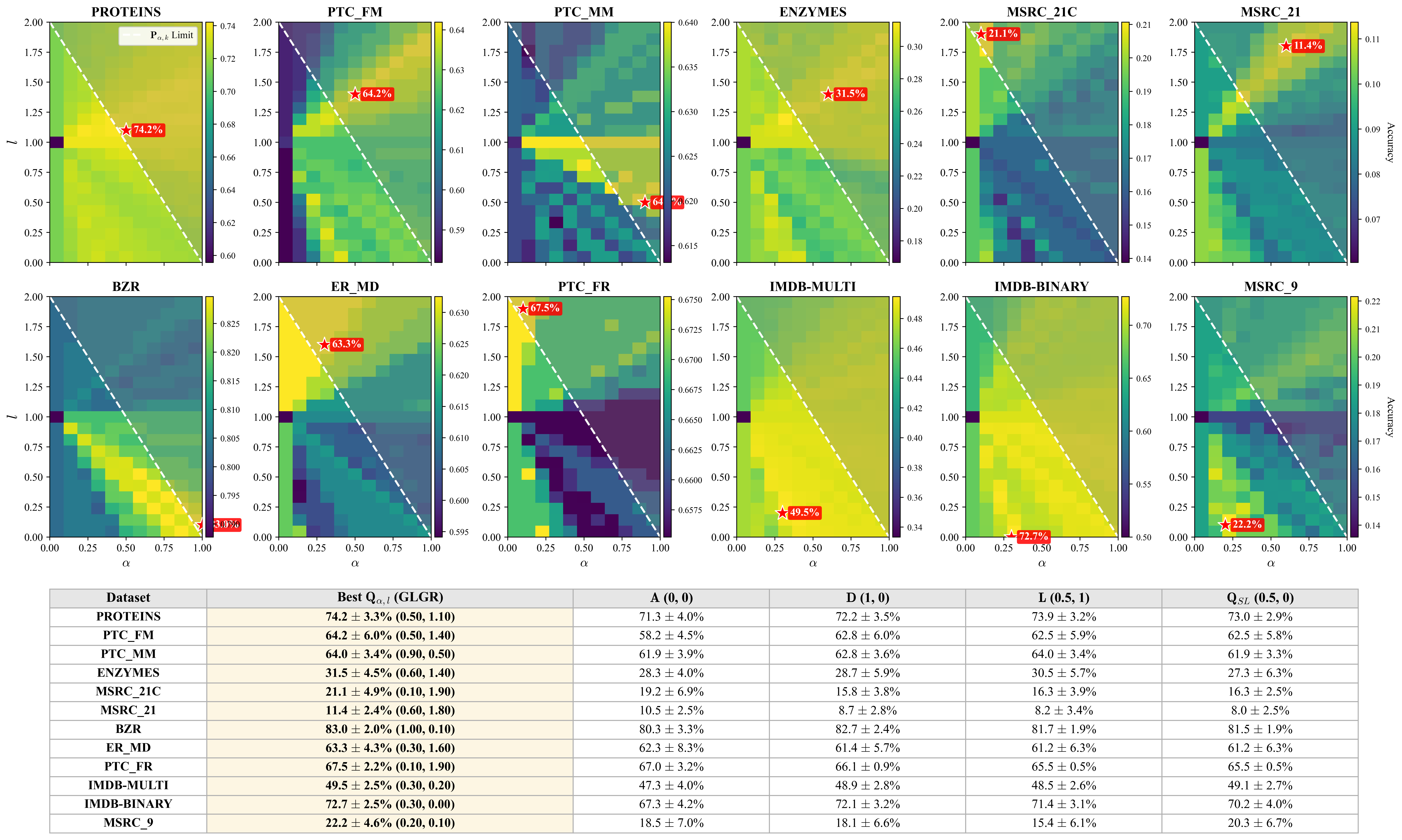}
    \caption{\textbf{Classification performance of GLGR across 12 datasets.} Heatmaps display the accuracy landscapes in the $(\alpha, l)$ space. The white dashed line marks the theoretical upper bound of the prior $\mathbf{P}_{\alpha,k}$ framework, highlighting the GLGR-exclusive zone above it. Red stars ($\star$) indicate the global optima. The bottom table compares GLGR with classical operators, showing strong performance by effectively leveraging negative topological weights within both shared and exclusive parameter domains.}
    \label{fig:glgr_comprehensive_results}
\end{figure*}

\begin{algorithm}[t]
\caption{GLGR Spectral Correlation (SCor) Grid Evaluation}
\label{alg:glgr_search}
\begin{algorithmic}[1]
\REQUIRE Dataset $\mathcal{D}=\{(G_i,y_i)\}_{i=1}^{m}$, parameter grids $\mathcal{A}$ and $\mathcal{L}$, fixed SVM settings $(C\ \text{untuned},\ \gamma=1)$
\ENSURE Accuracy map $\mathrm{Acc}(\alpha,l)$ and best pair $(\alpha^{\star},l^{\star})$
\FOR{each $(\alpha,l)\in\mathcal{A}\times\mathcal{L}$}
    \STATE Construct $\mathbf{Q}_{\alpha,l}(G_i)$ for each graph $G_i$
    \STATE Compute standardized spectra $\tilde{\nu}^{(\alpha,l)}(G_i)$
    \STATE Form kernel matrix $[\mathbf{K}_{\alpha,l}]_{ij}=\exp\!\left(-\gamma\,\mathrm{SCor}_{\alpha,l}(G_i,G_j)\right)$
    \STATE Train and evaluate SVM via one shuffled stratified 10-fold cross-validation split
    \STATE Store mean score in $\mathrm{Acc}(\alpha,l)$
\ENDFOR
\STATE $(\alpha^{\star},l^{\star})\leftarrow \arg\max_{(\alpha,l)}\mathrm{Acc}(\alpha,l)$
\RETURN $\mathrm{Acc}(\alpha,l)$, $(\alpha^{\star},l^{\star})$, and $\max_{(\alpha,l)}\mathrm{Acc}(\alpha,l)$
\end{algorithmic}
\end{algorithm}

\subsection{Graph Classification via SCor+SVM: Representation-Space Validation}
This experiment validates GLGR at the \emph{representation level}, independently of deep model architecture. Following prior spectral-correlation pipelines, we pair a spectral similarity kernel with SVM classification, and evaluate how performance varies across the $(\alpha,l)$ operator plane.

\textbf{Datasets.} We use 12 benchmark graph classification datasets from the standard TUDataset collection \cite{morris2020tudataset}, spanning bio-chemical and social-network domains: PROTEINS, PTC\_FM, PTC\_MM, ENZYMES, MSRC\_21C, MSRC\_21, BZR, ER\_MD, PTC\_FR, IMDB-MULTI, IMDB-BINARY, and MSRC\_9. Their accuracy landscapes are summarized in Fig.~\ref{fig:glgr_comprehensive_results}.

\textbf{SCor similarity under GLGR.} For graph $G$, let
\begin{equation}
\mathbf{Q}_{\alpha,l}(G)=\alpha\mathbf{D}(G)+(1-\alpha-l)\mathbf{A}(G),
\end{equation}
and let $\nu^{(\alpha,l)}(G)\in\mathbb{R}^{n_{\max}}$ be its eigenvalue vector after zero-padding/truncation to a common size $n_{\max}$, followed by standardization:
\begin{equation}
\tilde{\nu}^{(\alpha,l)}(G)=\frac{\nu^{(\alpha,l)}(G)-m_G\mathbf{1}}{s_G+\varepsilon}.
\end{equation}
Here $m_G$ and $s_G$ denote the mean and standard deviation of $\nu^{(\alpha,l)}(G)$, respectively, and $\mathbf{1}$ is the all-ones vector.
For graphs $G_i,G_j$, the GLGR spectral-correlation distance is
\begin{equation}
\mathrm{SCor}_{\alpha,l}(G_i,G_j)=\sqrt{1-\left(\frac{1}{n_{\max}}\left\langle \tilde{\nu}^{(\alpha,l)}(G_i),\tilde{\nu}^{(\alpha,l)}(G_j)\right\rangle\right)^2}.
\end{equation}
The SVM Gram matrix is then defined by an RBF-type map
\begin{equation}
[\mathbf{K}_{\alpha,l}]_{ij}=\exp\big(-\gamma\,\mathrm{SCor}_{\alpha,l}(G_i,G_j)\big),
\end{equation}
where $\gamma>0$ controls kernel sharpness. In our experiments, we fix $\gamma=1$ for all datasets.

\textbf{Protocol.} We perform uniform grid search with $\alpha\in[0,1]$, $l\in[0,2]$, step $0.1$. For each pair $(\alpha,l)$, we compute $\mathbf{K}_{\alpha,l}$ and evaluate an SVM using a single shuffled stratified 10-fold cross-validation split (implemented by StratifiedKFold with $\texttt{n\_splits}=10$, $\texttt{shuffle}=\texttt{True}$, and $\texttt{random\_state}=42$). The reported score is the mean accuracy across the 10 folds. To isolate the effect of the GLGR representation space, we do not tune the SVM penalty parameter $C$, and we keep $\gamma=1$ fixed for all datasets.
The full evaluation workflow is summarized in Algorithm~\ref{alg:glgr_search}.

\begin{table*}[t]
\centering
\caption{Statistics of the 13 node-classification datasets used in our experiments.}
\label{tab:node_dataset_stats}
{\footnotesize
\setlength{\tabcolsep}{3.5pt}
\renewcommand{\arraystretch}{1.05}
\resizebox{\textwidth}{!}{%
\begin{tabular}{l|*{13}{c}}
\toprule
Dataset & Cora & Citeseer & PubMed & Photo & Computers & Physics & WikiCS & Texas & Wisconsin & Cornell & Actor & Chameleon & Squirrel \\
\midrule
Type & Hom. & Hom. & Hom. & Hom. & Hom. & Hom. & Hom. & Het. & Het. & Het. & Het. & Het. & Het. \\
\#Nodes & 2,708 & 3,327 & 19,717 & 7,650 & 13,752 & 34,493 & 11,701 & 183 & 251 & 183 & 7,600 & 2,277 & 5,201 \\
\#Edges & 5,278 & 4,552 & 44,324 & 119,081 & 245,861 & 247,962 & 216,123 & 279 & 466 & 277 & 26,659 & 31,371 & 198,353 \\
\#Features & 1,433 & 3,703 & 500 & 745 & 767 & 8,415 & 300 & 1,703 & 1,703 & 1,703 & 932 & 2,325 & 2,089 \\
\#Classes & 7 & 6 & 3 & 8 & 10 & 5 & 10 & 5 & 5 & 5 & 5 & 5 & 5 \\
\bottomrule
\end{tabular}}
}
\end{table*}

The resulting maps support three main observations. First, the best points often lie away from the classical corners (pure $\mathbf{A}$/$\mathbf{L}$), which indicates that intermediate GLGR operators capture additional discriminative information. Second, several datasets attain their optimum in GLGR-exclusive regions beyond the compact prior $\mathbf{P}_{\alpha,k}$ boundary (white dashed line in Fig.~\ref{fig:glgr_comprehensive_results}), so the improvement cannot be reduced to a simple reparameterization effect. Third, high-accuracy regions are usually broad rather than spiky, which is consistent with the bounded linear parameterization and with stable operator selection. Together, these results show that GLGR already enlarges the useful spectral search space before end-to-end learning is introduced.

% homophilous datasets
\begin{table*}[t]
\centering
\caption{Node classification results on homophilous datasets: mean accuracy (\%) $\pm$ std. Within each baseline/GLGR backbone pair, the better result is shown in \textbf{bold}.}
\label{tab:node_homophilous}
\resizebox{\textwidth}{!}{
\begin{tabular}{c|ccccccc}
\toprule
Model & Cora & Citeseer & PubMed & Photo & Computers & Physics & WikiCS \\
\midrule
SGC & 81.19{\scriptsize $\pm$0.70} & 78.02{\scriptsize $\pm$0.73} & 84.99{\scriptsize $\pm$0.32} & 90.78{\scriptsize $\pm$1.60} & 82.77{\scriptsize $\pm$0.93} & 93.32{\scriptsize $\pm$0.09} & 79.22{\scriptsize $\pm$0.47} \\
GLGR-SGC & \textbf{86.97}{\scriptsize $\pm$0.81} & \textbf{79.28}{\scriptsize $\pm$0.79} & \textbf{87.16}{\scriptsize $\pm$0.26} & \textbf{91.96}{\scriptsize $\pm$0.34} & \textbf{85.48}{\scriptsize $\pm$0.48} & \textbf{96.09}{\scriptsize $\pm$0.09} & \textbf{80.68}{\scriptsize $\pm$0.45} \\
\midrule
GCN & 86.81{\scriptsize $\pm$0.61} & 78.51{\scriptsize $\pm$0.41} & 86.03{\scriptsize $\pm$0.28} & 88.61{\scriptsize $\pm$0.74} & 83.45{\scriptsize $\pm$0.16} & 95.68{\scriptsize $\pm$0.11} & 81.53{\scriptsize $\pm$0.36} \\
GLGR-GCN & \textbf{86.98}{\scriptsize $\pm$0.57} & \textbf{79.73}{\scriptsize $\pm$0.38} & \textbf{88.52}{\scriptsize $\pm$0.25} & \textbf{94.22}{\scriptsize $\pm$0.32} & \textbf{89.63}{\scriptsize $\pm$0.29} & \textbf{96.83}{\scriptsize $\pm$0.11} & \textbf{82.60}{\scriptsize $\pm$0.56} \\
\midrule
GPRGNN & 87.93{\scriptsize $\pm$0.79} & 79.02{\scriptsize $\pm$0.51} & 87.77{\scriptsize $\pm$0.40} & 93.69{\scriptsize $\pm$0.36} & 87.19{\scriptsize $\pm$0.68} & 96.81{\scriptsize $\pm$0.08} & 81.31{\scriptsize $\pm$0.43}\\
GLGR-GPRGNN & \textbf{88.77}{\scriptsize $\pm$0.58} & \textbf{80.43}{\scriptsize $\pm$0.35} & \textbf{88.82}{\scriptsize $\pm$0.33} & \textbf{94.68}{\scriptsize $\pm$0.25} & \textbf{89.22}{\scriptsize $\pm$0.49} & \textbf{97.47}{\scriptsize $\pm$0.11} & \textbf{82.32}{\scriptsize $\pm$0.36} \\
\midrule
APPNP & 88.41{\scriptsize $\pm$0.85} & 80.60{\scriptsize $\pm$0.76} & 86.21{\scriptsize $\pm$0.31} & 88.67{\scriptsize $\pm$0.42} & 85.39{\scriptsize $\pm$0.29} & 96.08{\scriptsize $\pm$0.12} & 79.90{\scriptsize $\pm$0.46} \\
GLGR-APPNP & \textbf{88.80}{\scriptsize $\pm$0.69} & \textbf{80.94}{\scriptsize $\pm$0.70} & \textbf{88.75}{\scriptsize $\pm$0.32} & \textbf{94.55}{\scriptsize $\pm$0.32} & \textbf{88.26}{\scriptsize $\pm$0.35} & \textbf{97.12}{\scriptsize $\pm$0.11} & \textbf{81.51}{\scriptsize $\pm$0.31} \\
\midrule
CHEBNET & 86.51{\scriptsize $\pm$0.63} & 78.74{\scriptsize $\pm$0.46} & 88.31{\scriptsize $\pm$0.20} & 93.74{\scriptsize $\pm$0.30} & 87.84{\scriptsize $\pm$0.40} & 97.25{\scriptsize $\pm$0.78} & 81.38{\scriptsize $\pm$0.41} \\
GLGR-CHEBNET & \textbf{87.29}{\scriptsize $\pm$0.81} & \textbf{80.82}{\scriptsize $\pm$1.50} & \textbf{89.35}{\scriptsize $\pm$0.46} & \textbf{94.91}{\scriptsize $\pm$0.34} & \textbf{89.16}{\scriptsize $\pm$0.89} & \textbf{97.62}{\scriptsize $\pm$0.15} & \textbf{83.62}{\scriptsize $\pm$0.41} \\
\midrule
BERNNET & 88.10{\scriptsize $\pm$0.85} & 79.86{\scriptsize $\pm$0.85} & 86.50{\scriptsize $\pm$0.28} & 93.62{\scriptsize $\pm$0.28} & 86.78{\scriptsize $\pm$0.44} & 96.37{\scriptsize $\pm$0.13} & 80.16{\scriptsize $\pm$0.48} \\
GLGR-BERNNET & \textbf{88.26}{\scriptsize $\pm$0.82} & \textbf{80.04}{\scriptsize $\pm$0.65} & \textbf{88.94}{\scriptsize $\pm$0.19} & \textbf{94.69}{\scriptsize $\pm$0.24} & \textbf{88.93}{\scriptsize $\pm$0.40} & \textbf{97.72}{\scriptsize $\pm$0.24} & \textbf{81.78}{\scriptsize $\pm$0.19} \\
\midrule
JACOBICONV & 87.89{\scriptsize $\pm$0.67} & 77.96{\scriptsize $\pm$0.45} & 87.37{\scriptsize $\pm$0.22} & 93.54{\scriptsize $\pm$0.30} & 87.18{\scriptsize $\pm$0.39} & 96.63{\scriptsize $\pm$0.11} & 80.92{\scriptsize $\pm$0.39} \\
GLGR-JACOBICONV & \textbf{88.34}{\scriptsize $\pm$0.60} & \textbf{78.98}{\scriptsize $\pm$0.53} & \textbf{89.78}{\scriptsize $\pm$0.20} & \textbf{95.20}{\scriptsize $\pm$0.29} & \textbf{90.37}{\scriptsize $\pm$0.32} & \textbf{97.95}{\scriptsize $\pm$0.10} & \textbf{81.98}{\scriptsize $\pm$0.37} \\
\midrule
SGFORMER & 85.59{\scriptsize $\pm$1.15} & 77.89{\scriptsize $\pm$0.94} & 87.02{\scriptsize $\pm$0.38} & 91.05{\scriptsize $\pm$0.32} & 84.99{\scriptsize $\pm$0.59} & 95.78{\scriptsize $\pm$0.23} & 81.69{\scriptsize $\pm$0.56} \\
GLGR-SGFORMER & \textbf{86.53}{\scriptsize $\pm$0.91} & \textbf{79.15}{\scriptsize $\pm$0.83} & \textbf{89.40}{\scriptsize $\pm$0.34} & \textbf{94.20}{\scriptsize $\pm$0.35} & \textbf{87.51}{\scriptsize $\pm$0.92} & \textbf{98.15}{\scriptsize $\pm$0.16} & \textbf{82.62}{\scriptsize $\pm$0.47} \\
\bottomrule
\end{tabular}}
\end{table*}

% heterophilous datasets
\begin{table*}[t!]
\centering
\caption{Node classification results on heterophilous datasets: mean accuracy (\%) $\pm$ std. Within each baseline/GLGR backbone pair, the better result is shown in \textbf{bold}.}
\label{tab:node_heterophilous}
\resizebox{\textwidth}{!}{
\begin{tabular}{c|cccccc}
\toprule
Model & Texas & Wisconsin & Cornell & Actor & Chameleon & Squirrel \\
\midrule
SGC & 83.29{\scriptsize $\pm$1.57} & 84.42{\scriptsize $\pm$1.03} & 81.46{\scriptsize $\pm$1.87} & 29.59{\scriptsize $\pm$2.65} & 41.49{\scriptsize $\pm$1.27} & 26.75{\scriptsize $\pm$0.66} \\
GLGR-SGC & \textbf{85.07}{\scriptsize $\pm$2.09} & \textbf{86.32}{\scriptsize $\pm$1.18} & \textbf{84.04}{\scriptsize $\pm$1.94} & \textbf{30.69}{\scriptsize $\pm$1.17} & \textbf{44.19}{\scriptsize $\pm$1.26} & \textbf{28.21}{\scriptsize $\pm$0.57} \\
\midrule
GCN & 77.34{\scriptsize $\pm$3.04} & 89.51{\scriptsize $\pm$3.14} & 79.45{\scriptsize $\pm$2.94} & 35.74{\scriptsize $\pm$0.78} & 58.62{\scriptsize $\pm$1.81} & 44.33{\scriptsize $\pm$0.61} \\
GLGR-GCN & \textbf{81.04}{\scriptsize $\pm$3.64} & \textbf{92.43}{\scriptsize $\pm$1.68} & \textbf{88.69}{\scriptsize $\pm$2.55} & \textbf{37.82}{\scriptsize $\pm$0.55} & \textbf{59.93}{\scriptsize $\pm$1.39} & \textbf{46.61}{\scriptsize $\pm$0.60} \\
\midrule
GPRGNN & 89.18{\scriptsize $\pm$1.79} & 79.09{\scriptsize $\pm$3.20} & 86.33{\scriptsize $\pm$3.10} & 39.67{\scriptsize $\pm$0.95} & 59.24{\scriptsize $\pm$1.92} & 46.30{\scriptsize $\pm$1.16} \\
GLGR-GPRGNN & \textbf{91.59}{\scriptsize $\pm$1.41} & \textbf{84.57}{\scriptsize $\pm$3.70} & \textbf{88.97}{\scriptsize $\pm$1.95} & \textbf{40.80}{\scriptsize $\pm$0.52} & \textbf{62.80}{\scriptsize $\pm$0.72} & \textbf{48.92}{\scriptsize $\pm$1.12} \\
\midrule
APPNP & 86.39{\scriptsize $\pm$2.46} & 87.18{\scriptsize $\pm$3.68} & 89.34{\scriptsize $\pm$1.96} & 40.02{\scriptsize $\pm$0.70} & 49.75{\scriptsize $\pm$1.31} & 32.34{\scriptsize $\pm$0.76} \\
GLGR-APPNP & \textbf{90.66}{\scriptsize $\pm$2.45} & \textbf{92.42}{\scriptsize $\pm$1.33} & \textbf{91.31}{\scriptsize $\pm$1.31} & \textbf{40.66}{\scriptsize $\pm$0.74} & \textbf{50.50}{\scriptsize $\pm$0.65} & \textbf{32.77}{\scriptsize $\pm$0.55} \\
\midrule
CHEBNET & 85.79{\scriptsize $\pm$2.89} & 88.00{\scriptsize $\pm$2.43} & 82.38{\scriptsize $\pm$2.65} & 37.32{\scriptsize $\pm$0.53} & 58.09{\scriptsize $\pm$0.91} & 39.70{\scriptsize $\pm$1.21} \\
GLGR-CHEBNET & \textbf{88.77}{\scriptsize $\pm$4.19} & \textbf{89.58}{\scriptsize $\pm$4.27} & \textbf{84.27}{\scriptsize $\pm$3.64} & \textbf{38.92}{\scriptsize $\pm$0.65} & \textbf{64.33}{\scriptsize $\pm$1.66} & \textbf{54.93}{\scriptsize $\pm$1.31} \\
\midrule
BERNNET & 91.96{\scriptsize $\pm$2.42} & 91.78{\scriptsize $\pm$1.98} & 88.85{\scriptsize $\pm$2.29} & 40.23{\scriptsize $\pm$0.77} & 60.31{\scriptsize $\pm$1.32} & 44.97{\scriptsize $\pm$0.97} \\
GLGR-BERNNET & \textbf{93.36}{\scriptsize $\pm$1.36} & \textbf{92.66}{\scriptsize $\pm$0.89} & \textbf{89.84}{\scriptsize $\pm$1.96} & \textbf{40.97}{\scriptsize $\pm$0.81} & \textbf{63.36}{\scriptsize $\pm$1.62} & \textbf{49.00}{\scriptsize $\pm$3.56} \\
\midrule
JACOBICONV & 83.95{\scriptsize $\pm$5.18} & 87.46{\scriptsize $\pm$2.57} & 86.80{\scriptsize $\pm$2.84} & 37.06{\scriptsize $\pm$0.74} & 63.77{\scriptsize $\pm$1.25} & 51.26{\scriptsize $\pm$2.92} \\
GLGR-JACOBICONV & \textbf{88.45}{\scriptsize $\pm$2.24} & \textbf{89.44}{\scriptsize $\pm$2.29} & \textbf{88.17}{\scriptsize $\pm$1.79} & \textbf{40.15}{\scriptsize $\pm$0.44} & \textbf{68.33}{\scriptsize $\pm$1.03} & \textbf{55.04}{\scriptsize $\pm$1.27} \\
\midrule
SGFORMER & 90.60{\scriptsize $\pm$2.07} & 92.89{\scriptsize $\pm$2.23} & 91.31{\scriptsize $\pm$2.30} & 37.58{\scriptsize $\pm$1.01} & 47.99{\scriptsize $\pm$1.49} & 31.07{\scriptsize $\pm$0.52} \\
GLGR-SGFORMER & \textbf{93.04}{\scriptsize $\pm$1.33} & \textbf{94.55}{\scriptsize $\pm$1.84} & \textbf{92.91}{\scriptsize $\pm$2.71} & \textbf{38.37}{\scriptsize $\pm$1.15} & \textbf{55.44}{\scriptsize $\pm$1.06} & \textbf{32.58}{\scriptsize $\pm$1.24} \\
\bottomrule
\end{tabular}}
\end{table*}

\subsection{Adaptive GLGR-GNN: Learnability and Generalization}
We next evaluate whether the representation-space advantage of GLGR can be realized through end-to-end learning. In this setting, $(\alpha,l)$ are optimized jointly with network parameters, and GLGR is injected into multiple backbones following Section V.

\textbf{Dataset Description.} We conduct node classification experiments on 13 widely used benchmarks, including seven homophilous datasets (Cora, Citeseer, and PubMed \cite{yang2016revisiting}; Photo, Computers, and Physics \cite{shchur2018pitfalls}; and WikiCS \cite{mernyei2020wikics}) and six heterophilous datasets (Texas, Wisconsin, Cornell, Actor, Chameleon, and Squirrel \cite{pei2020geom}). This split allows us to evaluate whether adaptive GLGR can simultaneously handle smooth low-frequency propagation and heterophily-driven high-frequency discrimination. Table~\ref{tab:node_dataset_stats} summarizes the benchmark versions used in our experiments, including graph scale, attribute dimension and class count. The wide spread in both graph size and homophily makes this suite a suitable testbed for checking whether a learnable propagation operator can adapt across markedly different structural regimes rather than overfit to a single benchmark family.

\textbf{Baselines and Settings.} We consider eight representative backbones: SGC \cite{wu2019simplifying}, GCN \cite{kipf2017semi}, GPRGNN \cite{chien2021adaptive}, APPNP \cite{klicpera2019predict}, CHEBNET \cite{defferrard2016convolutional}, BERNNET \cite{he2021bernnet}, JACOBICONV \cite{wang2022powerful}, and SGFORMER \cite{wu2023sgformer}. For each backbone, the GLGR-enhanced variant is obtained by replacing only the original fixed propagation operator with the adaptive GLGR instantiation summarized in Table~\ref{tab:backbone_instantiation}, while preserving the corresponding architecture, depth, optimizer choice, and training schedule.

\begin{figure*}[t]
    \centering
    \includegraphics[width=\textwidth]{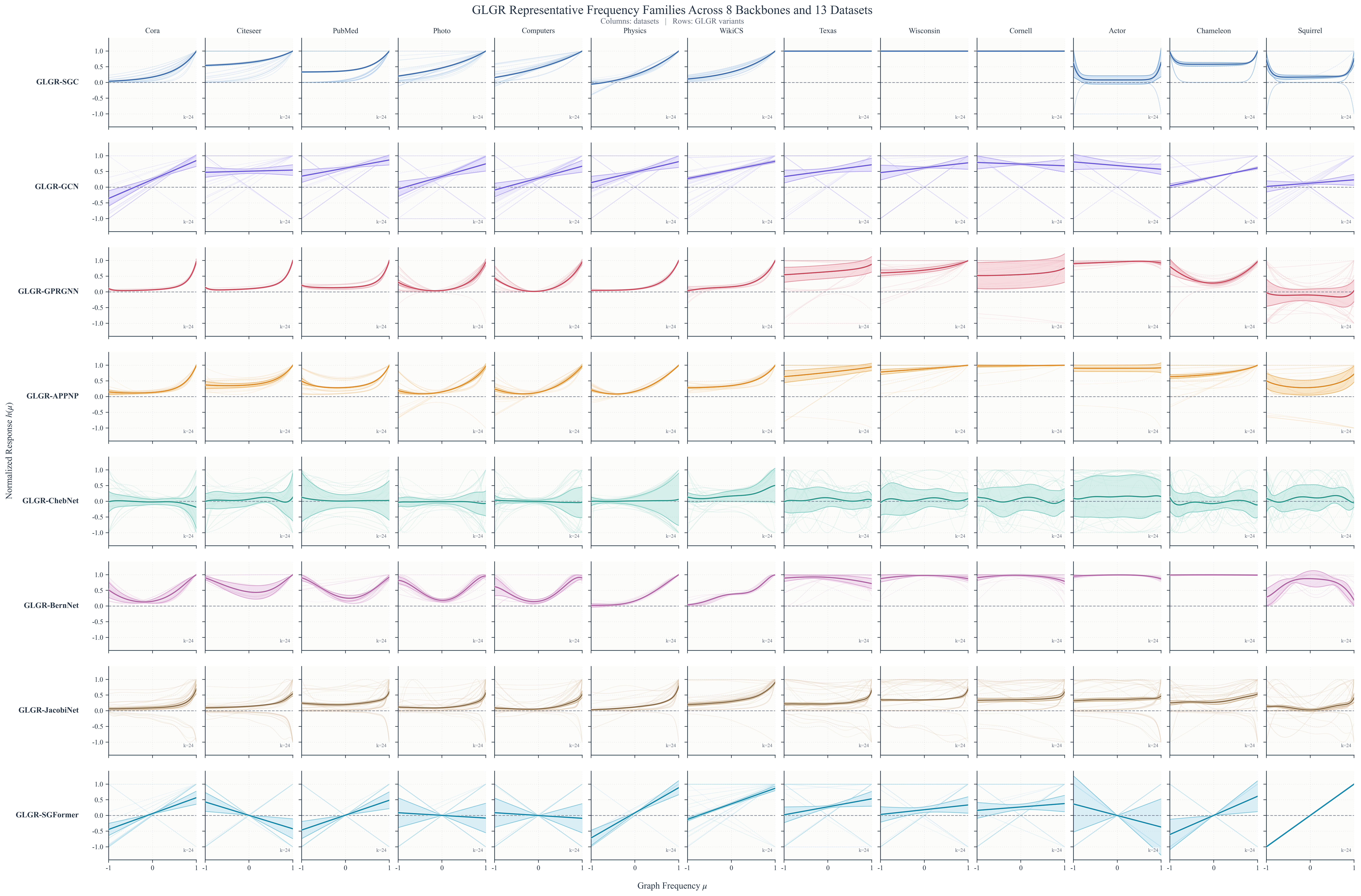}
    \caption{\textbf{Cross-dataset summary of learned GLGR spectral-response curves.} Rows denote GLGR-enhanced backbones and columns denote datasets. In each cell, thin curves show sampled component-level responses, the thick curve shows the mean response, and the shaded band marks one standard deviation across runs.}
    \label{fig:glgr_frequency_grid}
\end{figure*}

To ensure a fair comparison, we do not perform any additional backbone-specific hyperparameter search in our implementation. Instead, for every backbone we directly adopt the default hyperparameter settings provided in the public BernNet codebase and keep them fixed for both the original baseline and its GLGR-enhanced counterpart. In particular, the learning rate, weight decay, dropout, hidden width, depth/order, propagation steps, early-stopping strategy, and all other architecture-specific settings are kept identical between each baseline and the corresponding GLGR variant. The GLGR-enhanced model introduces only the structural scalars $(a^{(k)},b^{(k)})$---equivalently $(\alpha^{(k)},l^{(k)})$---while the remaining model parameters and training settings are inherited unchanged from the matched baseline. All experiments are implemented in Python 3.9.12 with PyTorch 1.11.0+cu113 and are run on a workstation equipped with a single NVIDIA RTX 3060 GPU and a 12th Gen Intel(R) Core(TM) i5-12400F CPU. We split nodes into train/validation/test sets with a $60\%/20\%/20\%$ ratio and generate 10 random splits; all methods are evaluated on exactly the same 10 splits for strict comparability. Results are reported as mean classification accuracy $\pm$ standard deviation across the 10 runs. For polynomial GNNs, we fix the order to $K=10$ following common practice. Consequently, the reported gains compare each GLGR-enhanced model against its corresponding baseline under exactly the same default hyperparameter configuration, rather than under separately tuned settings.

\textbf{Results and Discussion.} The pairwise comparisons in Tables~\ref{tab:node_homophilous} and \ref{tab:node_heterophilous} show consistent gains of GLGR-enhanced variants over their original baselines. On homophilous datasets, improvements are stable and moderate, indicating that GLGR preserves beneficial smoothing while refining spectral selectivity. On heterophilous datasets, gains are typically larger (especially on Texas, Wisconsin, and Cornell), where fixed low-pass propagation is known to be suboptimal \cite{nt2019revisiting,zhu2020beyond,chien2021adaptive}. These observations agree with the analysis in Section V: adaptive learning of $(\alpha,l)$ enables data-dependent frequency responses, allowing each backbone to move away from one-size-fits-all smoothing. Importantly, these gains are obtained under the matched default-setting protocol above, so they cannot be attributed to extra hyperparameter tuning or to backbone redesign outside the propagation operator. For backbones whose original formulations are tied to specific normalization, diffusion, or polynomial-approximation assumptions, the present comparisons should therefore be interpreted in the operator-injection sense of Section~V: GLGR acts as a matched empirical replacement of the propagation basis within an otherwise unchanged architecture, rather than as a claim that the parent model's original operator-specific semantics are preserved verbatim. Overall, this experiment confirms both \emph{learnability} and \emph{generalization}: GLGR is trainable in end-to-end settings and transferable across heterogeneous model families.

\subsection{Cross-Dataset Frequency-Response Summary}
To complement the accuracy tables with a direct view of the learned propagation behavior, Fig.~\ref{fig:glgr_frequency_grid} presents a cross-dataset summary of the learned GLGR spectral-response curves across all eight GLGR-enhanced node-classification backbones and thirteen datasets. Each cell corresponds to one backbone--dataset pair. For visual readability, the thin curves show a small sampled subset of component-level responses rescaled for visualization, the thick curve shows the mean overall response, and the shaded band indicates one standard deviation across runs. The horizontal axis uses the shared normalized coordinate $\bar{\mu}\in[-1,1]$, so all responses are displayed on the same GLGR scale. The figure gives a compact qualitative summary of how GLGR shapes spectral preference under different model and data conditions.

Fig.~\ref{fig:glgr_frequency_grid} suggests that the learned GLGR responses are not concentrated around a single universal spectral pattern. Across many backbone--dataset pairs, the curves occupy related but distinguishable spectral regions, indicating that the learned operator can adjust its propagation preference to different model and data contexts rather than reverting to one fixed bias. In a large portion of the grid, the sampled component-level curves remain centered around a relatively smooth mean trajectory and the uncertainty bands stay within a moderate range, which is consistent with generally stable optimization in the compact GLGR space. Meanwhile, the width and shape of the response family still vary across backbone--dataset pairs, suggesting that GLGR retains enough flexibility to reflect graph-dependent structural differences while remaining within a compact and interpretable operator space. Taken together, the figure provides a qualitative complement to the quantitative results by showing that GLGR offers a compact operator space in which different backbones can learn different spectral preferences across datasets.

\begin{table*}[t]
\centering
\caption{Unified node-classification accuracy comparison (\%) on seven homophilous and six heterophilous benchmarks. The last row reports the strongest GLGR-enhanced result on each dataset, and the superscript indicates the corresponding backbone.}
\label{tab:unified_large_comparison}

{\footnotesize
\setlength{\tabcolsep}{2pt}          
\renewcommand{\arraystretch}{1.05}
\newcommand{\best}[1]{\textbf{#1}}
\newcommand{\second}[1]{\underline{#1}}

\resizebox{\textwidth}{!}{       
\begin{tabular}{c|ccccccc|cccccc}
\toprule
\multirow{2}{*}{Method} & \multicolumn{7}{c|}{Homophilous Datasets} & \multicolumn{6}{c}{Heterophilous Datasets} \\
 & Cora & Citeseer & PubMed & Photo & Computers & Physics & WikiCS & Texas & Wisconsin & Cornell & Actor & Chameleon & Squirrel \\
\midrule

MLP & 75.92{\tiny $\pm$0.95} & 76.34{\tiny $\pm$1.04} & 84.34{\tiny $\pm$0.41} & 84.98{\tiny $\pm$0.42} & 83.11{\tiny $\pm$0.21} & 94.55{\tiny $\pm$0.13} & 75.53{\tiny $\pm$0.45} & 88.38{\tiny $\pm$2.41} & \second{90.11}{\tiny $\pm$2.49} & \second{89.71}{\tiny $\pm$2.09} & 39.15{\tiny $\pm$0.64} & 47.64{\tiny $\pm$1.06} & 31.50{\tiny $\pm$0.62} \\

GAT & 88.19{\tiny $\pm$0.60} & 80.70{\tiny $\pm$0.67} & 85.56{\tiny $\pm$0.43} & 91.69{\tiny $\pm$0.66} & \second{84.21}{\tiny $\pm$0.53} & 95.95{\tiny $\pm$0.12} & 81.85{\tiny $\pm$0.29} & 75.64{\tiny $\pm$1.80} & 66.83{\tiny $\pm$1.97} & 66.92{\tiny $\pm$3.30} & 35.17{\tiny $\pm$0.41} & 59.96{\tiny $\pm$1.32} & 33.18{\tiny $\pm$0.99} \\

H2GCN & 87.96{\tiny $\pm$0.37} & 80.90{\tiny $\pm$1.21} & 89.18{\tiny $\pm$0.28} & \second{95.45}{\tiny $\pm$0.67} & - & 97.19{\tiny $\pm$0.13} & 83.45{\tiny $\pm$0.26} & \second{91.89}{\tiny $\pm$3.93} & 86.67{\tiny $\pm$4.69} & 84.05{\tiny $\pm$5.30} & 36.31{\tiny $\pm$2.58} & 61.20{\tiny $\pm$4.28} & 39.53{\tiny $\pm$0.88} \\

HopGNN & \second{88.68}{\tiny $\pm$1.06} & 80.38{\tiny $\pm$0.68} & 89.15{\tiny $\pm$0.35} & 94.49{\tiny $\pm$0.33} & - & 97.86{\tiny $\pm$0.16} & \best{84.73}{\tiny $\pm$0.59} & 89.15{\tiny $\pm$4.04} & 85.69{\tiny $\pm$5.43} & 83.51{\tiny $\pm$5.19} & 39.33{\tiny $\pm$2.79} & \second{65.25}{\tiny $\pm$3.49} & \best{57.83}{\tiny $\pm$2.11} \\

Transformer & 71.83{\tiny $\pm$1.68} & 70.55{\tiny $\pm$1.20} & 86.66{\tiny $\pm$0.50} & 89.58{\tiny $\pm$1.05} & - & OOM & 77.36{\tiny $\pm$1.25} & 88.75{\tiny $\pm$6.30} & - & - & 39.95{\tiny $\pm$0.64} & 45.21{\tiny $\pm$2.01} & 33.17{\tiny $\pm$1.32} \\

GraphGPS & 83.42{\tiny $\pm$1.22} & 75.87{\tiny $\pm$0.71} & 86.62{\tiny $\pm$0.53} & 94.35{\tiny $\pm$0.25} & - & 97.60{\tiny $\pm$0.05} & 79.26{\tiny $\pm$0.57} & 83.71{\tiny $\pm$5.85} & - & - & 37.68{\tiny $\pm$0.94} & 46.07{\tiny $\pm$1.51} & 34.14{\tiny $\pm$0.73} \\

NodeFormer & 87.32{\tiny $\pm$0.92} & 79.56{\tiny $\pm$1.10} & 89.24{\tiny $\pm$0.23} & 95.27{\tiny $\pm$0.22} & - & 96.45{\tiny $\pm$0.28} & 81.03{\tiny $\pm$0.94} & 84.63{\tiny $\pm$3.47} & - & - & 35.17{\tiny $\pm$0.41} & 56.34{\tiny $\pm$1.11} & 43.42{\tiny $\pm$1.62} \\

NAGphormer & 88.15{\tiny $\pm$1.35} & 80.12{\tiny $\pm$1.24} & 89.70{\tiny $\pm$0.19} & \best{95.49}{\tiny $\pm$0.11} & - & 97.85{\tiny $\pm$0.26} & 83.41{\tiny $\pm$0.34} & 91.80{\tiny $\pm$1.85} & - & - & 40.08{\tiny $\pm$1.50} & 54.92{\tiny $\pm$1.11} & 48.55{\tiny $\pm$2.56} \\

PolyFormer & 87.67{\tiny $\pm$1.28} & \best{81.80}{\tiny $\pm$0.76} & \best{90.68}{\tiny $\pm$0.31} & 94.08{\tiny $\pm$1.37} & - & \second{98.08}{\tiny $\pm$0.27} & \second{83.62}{\tiny $\pm$0.17} & 89.02{\tiny $\pm$5.44} & - & - & \best{41.51}{\tiny $\pm$0.71} & 60.17{\tiny $\pm$1.39} & 44.98{\tiny $\pm$3.03} \\

\midrule
\textbf{Ours}{\tiny (Best GLGR)}
& \best{88.80}{\tiny $\pm$0.69}$^{\mathrm{A}}$
& \second{80.94}{\tiny $\pm$0.70}$^{\mathrm{A}}$
& \second{89.78}{\tiny $\pm$0.20}$^{\mathrm{B}}$
& 95.20{\tiny $\pm$0.29}$^{\mathrm{B}}$
& \best{90.37}{\tiny $\pm$0.32}$^{\mathrm{B}}$
& \best{98.15}{\tiny $\pm$0.16}$^{\mathrm{C}}$
& \second{83.62}{\tiny $\pm$0.41}$^{\mathrm{D}}$
& \best{93.36}{\tiny $\pm$1.36}$^{\mathrm{E}}$
& \best{94.55}{\tiny $\pm$1.84}$^{\mathrm{C}}$
& \best{92.91}{\tiny $\pm$2.71}$^{\mathrm{C}}$
& \second{40.97}{\tiny $\pm$0.81}$^{\mathrm{E}}$
& \best{68.33}{\tiny $\pm$1.03}$^{\mathrm{B}}$
& \second{55.04}{\tiny $\pm$1.27}$^{\mathrm{B}}$ \\
\bottomrule
\end{tabular}
}
}

\vspace{3pt}

{\footnotesize
\textbf{Note:} Bold and underlined entries denote the best and second-best mean accuracy in each column. ``-'' indicates unavailable results, and ``OOM'' indicates out-of-memory. \textbf{GLGR code:} $\mathrm{A}$=GLGR-APPNP; $\mathrm{B}$=GLGR-JACOBICONV; $\mathrm{C}$=GLGR-SGFORMER; $\mathrm{D}$=GLGR-CHEBNET; $\mathrm{E}$=GLGR-BERNNET.}
\end{table*}

\subsection{Comparison with Other Representative Methods}
As a broader contextual reference, Table~\ref{tab:unified_large_comparison} compares representative non-GLGR methods, including MLP \cite{lecun1998gradient}, GAT \cite{velivckovic2017graph}, H2GCN \cite{zhu2020beyond}, HopGNN \cite{chen2023hopgnn}, Transformer \cite{dwivedi2021generalization}, GraphGPS \cite{rampasek2022recipe}, NodeFormer \cite{wu2022nodeformer}, NAGphormer \cite{chen2023nagphormer}, and PolyFormer \cite{ma2024polyformer}, with the strongest GLGR-enhanced result obtained on each dataset; the superscript in the GLGR row identifies the corresponding backbone.

Because these methods come from different architectural families and experimental settings, we use this table only as supplementary evidence. Even so, the strongest GLGR-enhanced result is best or second-best on most datasets and remains close to the leading result on the others, which indicates that the proposed operator stays competitive beyond matched backbone-level comparisons.

\section{Conclusion}

In this paper, we introduced the GLGR ($\mathbf{Q}_{\alpha, l}$), a compact linear operator family for graph representation in GSP and GNNs. By replacing the non-compact parameter coupling in prior formulations with a bounded two-parameter space, GLGR brings major classical operators and transition-type operators into the same learnable domain. We showed that this formulation is not only algebraically unified but also theoretically interpretable: the associated quadratic form yields a variational decomposition between smoothness and global energy, the spectrum remains stably controlled under graph perturbations, and graph-aware sufficient conditions can be derived for positive semi-definiteness. We further developed AG-Conv, which learns the propagation operator itself within end-to-end graph models. Experiments at both the representation level and the learning level show that GLGR improves operator selection and enhances multiple GNN backbones across graph classification and node classification benchmarks. More broadly, GLGR reframes graph operator design as optimization over a compact and analyzable operator space, providing a principled connection between spectral graph theory and adaptive graph learning. Future work will study directed and temporal extensions, together with richer analyses of learned operator trajectories and spectral responses.

\balance
\bibliographystyle{IEEEtran}
\bibliography{ref}

@article{ortega2018graph,
  title={Graph signal processing: {O}verview, challenges, and applications},
  author={Ortega, Antonio and Frossard, Pascal and Kova{\v{c}}evi{\'c}, Jelena and Moura, Jos{\'e} MF and Vandergheynst, Pierre},
  journal={Proc. IEEE},
  volume={106},
  number={5},
  pages={808--828},
  year={May 2018},
  publisher={IEEE}
}

@article{shuman2013emerging,
  title={The emerging field of signal processing on graphs: {E}xtending high-dimensional data analysis to networks and other irregular domains},
  author={Shuman, David I and Narang, Sunil K and Frossard, Pascal and Ortega, Antonio and Vandergheynst, Pierre},
  journal={IEEE Signal Process. Mag.},
  volume={30},
  number={3},
  pages={83--98},
  year={May 2013},
  publisher={IEEE}
}

@article{Sandryhaila2013,
  title={Discrete signal processing on graphs},
  author={Sandryhaila, Aliaksei and Moura, Jos{\'e} M.~F.},
  journal={IEEE Trans. Signal Process.},
  volume={61},
  number={7},
  pages={1644--1656},
  year={Apr. 2013},
  publisher={IEEE}
}

@article{Nikiforov2017,
  title={Merging the {A}- and {Q}-spectral theories},
  author={Nikiforov, Vladimir},
  journal={Appl. Anal. Discrete Math.},
  volume={11},
  number={1},
  pages={81--107},
  year={2017}
}

@article{Wang2020,
  title={Bounds for the largest and the smallest {$A_{\alpha}$} eigenvalues of a graph in terms of vertex degrees},
  author={Wang, Sai and Wong, Dein and Tian, Fenglei},
  journal={Linear Algebra Appl.},
  volume={590},
  pages={210--223},
  year={2020},
  publisher={Elsevier}
}

@article{Averty2024,
  title={A New Family of Graph Representation Matrices: Application to Graph and Signal Classification},
  author={Averty, T and Boudraa, AO and Dar{\'e}-Emzivat, D},
  journal={IEEE Signal Process. Lett.},
  volume={31},
  pages={2935--2939},
  year={Oct. 2024},
  publisher={IEEE}
}

@article{isufi2024graph,
  title={Graph filters for signal processing and machine learning on graphs},
  author={Isufi, Elvin and Gama, Fernando and Shuman, David I and Segarra, Santiago},
  journal={IEEE Trans. Signal Process.},
  volume={72},
  pages={4745--4781},
  year={2024},
  publisher={IEEE}
}

@article{morgenstern2024theoretical,
  title={Theoretical guarantees for sparse graph signal recovery},
  author={Morgenstern, Gal and Routtenberg, Tirza},
  journal={IEEE Signal Process. Lett.},
  volume={32},
  pages={266--270},
  year={2025},
  publisher={IEEE}
}

@article{buchnik2024gsp,
  title={{GSP-KalmanNet}: {T}racking graph signals via neural-aided {K}alman filtering},
  author={Buchnik, Itay and Sagi, Guy and Leinwand, Nimrod and Loya, Yuval and Shlezinger, Nir and Routtenberg, Tirza},
  journal={IEEE Trans. Signal Process.},
  volume={72},
  pages={3700--3716},
  year={2024},
  publisher={IEEE}
}

@article{castro2024gegenbauer,
  title={{G}egenbauer graph neural networks for time-varying signal reconstruction},
  author={Castro-Correa, Jhon A and Giraldo, Jhony H and Badiey, Mohsen and Malliaros, Fragkiskos D},
  journal={IEEE Trans. Neural Netw. Learn. Syst.},
  volume={35},
  number={9},
  pages={11734--11745},
  year={Sep. 2024},
  publisher={IEEE}
}

@inproceedings{smola2003kernels,
  title={Kernels and regularization on graphs},
  author={Smola, Alexander J and Kondor, Risi},
  booktitle={Proc. Learn. Theory Kernel Mach. (COLT)},
  pages={144--158},
  year={2003},
  organization={Springer}
}

@article{wang2016trend,
  title={Trend filtering on graphs},
  author={Wang, Yu-Xiang and Sharpnack, James and Smola, Alexander J and Tibshirani, Ryan J},
  journal={J. Mach. Learn. Res.},
  volume={17},
  number={105},
  pages={1--41},
  year={2016}
}

@inproceedings{kalofolias2016learn,
  title={How to learn a graph from smooth signals},
  author={Kalofolias, Vassilis},
  booktitle={Proc. Int. Conf. Artif. Intell. Statist. (AISTATS)},
  pages={920--929},
  year={2016}
}

@article{dong2020graph,
  title={Graph signal processing for machine learning: {A} review and new perspectives},
  author={Dong, Xiaowen and Thanou, Dorina and Toni, Laura and Bronstein, Michael and Frossard, Pascal},
  journal={IEEE Signal Process. Mag.},
  volume={37},
  number={6},
  pages={117--127},
  year={Nov. 2020},
  publisher={IEEE}
}

@article{song2024graph,
  title={Graph signal processing based nonlinear {QSAR}/{QSPR} model learning for compounds},
  author={Song, Xiaoying and Wen, Gaoya and Chai, Li},
  journal={Biomed. Signal Process. Control},
  volume={91},
  pages={106011},
  year={2024},
  publisher={Elsevier}
}

@article{bronstein2017geometric,
  title={Geometric deep learning: {G}oing beyond {E}uclidean data},
  author={Bronstein, Michael M and Bruna, Joan and LeCun, Yann and Szlam, Arthur and Vandergheynst, Pierre},
  journal={IEEE Signal Process. Mag.},
  volume={34},
  number={4},
  pages={18--42},
  year={Jul. 2017},
  publisher={IEEE}
}

@inproceedings{kipf2017semi,
  title={Semi-supervised classification with graph convolutional networks},
  author={Kipf, Thomas N and Welling, Max},
  booktitle={Proc. Int. Conf. Learn. Represent. (ICLR)},
  year={2017}
}

@article{wu2020comprehensive,
  title={A comprehensive survey on graph neural networks},
  author={Wu, Zonghan and Pan, Shirui and Chen, Fengwen and Long, Guodong and Zhang, Chengqi and Yu, Philip S},
  journal={IEEE Trans. Neural Netw. Learn. Syst.},
  volume={32},
  number={1},
  pages={4--24},
  year={Jan. 2021},
  publisher={IEEE}
}

@inproceedings{li2018deeper,
  title={Deeper insights into graph convolutional networks for semi-supervised learning},
  author={Li, Qimai and Han, Zhichao and Wu, Xiao-Ming},
  booktitle={Proc. AAAI Conf. Artif. Intell.},
  volume={32},
  number={1},
  year={2018}
}

@inproceedings{oono2020graph,
  title={Graph neural networks exponentially lose expressive power for node classification},
  author={Oono, Kenta and Suzuki, Taiji},
  booktitle={Proc. Int. Conf. Learn. Represent. (ICLR)},
  year={2020}
}

@article{nt2019revisiting,
  title={Revisiting graph neural networks: {A}ll we have is low-pass filters},
  author={Nt, Hoang and Maehara, Takanori},
  journal={arXiv preprint arXiv:1905.09550},
  year={2019}
}

@inproceedings{dwivedi2021generalization,
  title={A generalization of transformer networks to graphs},
  author={Dwivedi, Vijay Prakash and Bresson, Xavier},
  booktitle={AAAI Workshop on Deep Learning on Graphs: Methods and Applications},
  year={2021}
}

@inproceedings{ying2021do,
  title={Do {T}ransformers really perform badly for graph representation?},
  author={Ying, Chengxuan and Cai, Tianle and Luo, Shengjie and Zheng, Shuxin and Ke, Guolin and He, Di and Shen, Yanming and Liu, TieYan},
  booktitle={Proc. Adv. Neural Inf. Process. Syst. (NeurIPS)},
  volume={34},
  pages={28877--28888},
  year={2021}
}

@inproceedings{rampasek2022recipe,
  title={Recipe for a general, powerful, scalable graph {T}ransformer},
  author={Ramp{\'a}{\v{s}}ek, Ladislav and Galkin, Michael and Dwivedi, Vijay Prakash and Luu, Anh Tuan and Wolf, Guy and Beaini, Dominique},
  booktitle={Proc. Adv. Neural Inf. Process. Syst. (NeurIPS)},
  volume={35},
  pages={14501--14515},
  year={2022}
}

@inproceedings{wu2022nodeformer,
  title={{NodeFormer}: {A} scalable graph structure learning {T}ransformer for node classification},
  author={Wu, Qitian and Zhao, Wentao and Li, Zenan and Wipf, David P and Yan, Junchi},
  booktitle={Proc. Adv. Neural Inf. Process. Syst. (NeurIPS)},
  volume={35},
  pages={27387--27401},
  year={2022}
}

@article{wang2024graphmamba,
  title={{G}raph-{M}amba: {T}owards long-range graph sequence modeling with selective state spaces},
  author={Wang, Chloe and Tsepa, Oleksii and Ma, Jun and Wang, Bo},
  journal={arXiv preprint arXiv:2402.00789},
  year={2024}
}

@inproceedings{behrouz2024graphmamba,
  title={{G}raph {M}amba: {T}owards learning on graphs with state space models},
  author={Behrouz, Ali and Hashemi, Farnoosh},
  booktitle={Proc. ACM SIGKDD Int. Conf. Knowl. Discovery Data Mining},
  pages={119--130},
  year={2024}
}

@inproceedings{barnard1995spectral,
  title={A spectral algorithm for envelope reduction of sparse matrices},
  author={Barnard, Stephen T and Pothen, Alex and Simon, Horst D},
  booktitle={Proc. ACM/IEEE Conf. Supercomput.},
  pages={493--502},
  year={1993}
}

@article{belkin2003laplacian,
  title={Laplacian eigenmaps for dimensionality reduction and data representation},
  author={Belkin, Mikhail and Niyogi, Partha},
  journal={Neural Comput.},
  volume={15},
  number={6},
  pages={1373--1396},
  year={2003},
  publisher={MIT Press}
}

@article{kartal2021graph,
  title={Graph signal processing: {V}ertex multiplication},
  author={Kartal, Bunyamin and Bayiz, Yigit E and Ko{\c{c}}, Aykut},
  journal={IEEE Signal Process. Lett.},
  volume={28},
  pages={1270--1274},
  year={2021},
  publisher={IEEE}
}

@article{vandeville2017when,
  title={When {S}lepian meets {F}iedler: {P}utting a focus on the graph spectrum},
  author={Van De Ville, Dimitri and Demesmaeker, Robin and Preti, Maria Giulia},
  journal={IEEE Signal Process. Lett.},
  volume={24},
  number={7},
  pages={1001--1004},
  year={Jul. 2017},
  publisher={IEEE}
}

@article{dwivedi2023benchmarking,
  title={Benchmarking graph neural networks},
  author={Dwivedi, Vijay Prakash and Joshi, Chaitanya K and Luu, Anh Tuan and Laurent, Thomas and Bengio, Yoshua and Bresson, Xavier},
  journal={J. Mach. Learn. Res.},
  volume={24},
  number={43},
  pages={1--48},
  year={2023}
}

@inproceedings{velivckovic2017graph,
  title={Graph attention networks},
  author={Veli{\v{c}}kovi{\'c}, Petar and Cucurull, Guillem and Casanova, Arantxa and Romero, Adriana and Lio, Pietro and Bengio, Yoshua},
  booktitle={Proc. Int. Conf. Learn. Represent. (ICLR)},
  year={2018}
}

@inproceedings{wu2019simplifying,
  title={Simplifying graph convolutional networks},
  author={Wu, Felix and Zhang, Tianyi and Souza, Amauri and Fifty, Christopher and Yu, Tao and Weinberger, Kilian Q.},
  booktitle={Proc. Int. Conf. Mach. Learn. (ICML)},
  pages={6861--6871},
  year={2019}
}

@inproceedings{defferrard2016convolutional,
  title={Convolutional neural networks on graphs with fast localized spectral filtering},
  author={Defferrard, Micha{\"e}l and Bresson, Xavier and Vandergheynst, Pierre},
  booktitle={Proc. Adv. Neural Inf. Process. Syst. (NeurIPS)},
  volume={29},
  year={2016}
}

@inproceedings{klicpera2019predict,
  title={Predict then propagate: Graph neural networks meet personalized {P}age{R}ank},
  author={Gasteiger, Johannes and Bojchevski, Aleksandar and G{\"u}nnemann, Stephan},
  booktitle={Proc. Int. Conf. Learn. Represent. (ICLR)},
  year={2019}
}

@inproceedings{chien2021adaptive,
  title={Adaptive universal generalized {P}age{R}ank graph neural network},
  author={Chien, Eli and Peng, Jianhao and Li, Pan and Milenkovic, Olgica},
  booktitle={Proc. Int. Conf. Learn. Represent. (ICLR)},
  year={2021}
}

@inproceedings{he2021bernnet,
  title={{BernNet}: Learning arbitrary graph spectral filters via {B}ernstein approximation},
  author={He, Mingguo and Wei, Zhewei and Huang, Zengfeng and Xu, Hongteng},
  booktitle={Proc. Adv. Neural Inf. Process. Syst. (NeurIPS)},
  volume={34},
  pages={14239--14251},
  year={2021}
}

@inproceedings{wang2022powerful,
  title={How powerful are spectral graph neural networks},
  author={Wang, Xiyuan and Zhang, Muhan},
  booktitle={Proc. Int. Conf. Mach. Learn. (ICML)},
  pages={23341--23362},
  year={2022}
}

@inproceedings{wu2023sgformer,
  title={{SGFormer}: Simplifying and empowering transformers for large-graph representations},
  author={Wu, Qitian and Zhao, Wentao and Yang, Chenxiao and Zhang, Hengrui and Nie, Fan and Jiang, Haitian and Bian, Yatao and Yan, Junchi},
  booktitle={Proc. Adv. Neural Inf. Process. Syst. (NeurIPS)},
  volume={36},
  pages={64753--64773},
  year={2023}
}

@inproceedings{zhu2020beyond,
  title={Beyond homophily in graph neural networks: Current limitations and effective designs},
  author={Zhu, Jiong and Yan, Yanqiao and Zhao, Lingxiao and Heimann, Mark and Akoglu, Leman and Koutra, Danai},
  booktitle={Proc. Adv. Neural Inf. Process. Syst. (NeurIPS)},
  volume={33},
  pages={7793--7804},
  year={2020}
}

@inproceedings{chen2023hopgnn,
  title={From node interaction to hop interaction: New effective and scalable graph learning paradigm},
  author={Chen, Jie and Li, Zilong and Zhu, Yin and Zhang, Junping and Pu, Jian},
  booktitle={Proc. IEEE/CVF Conf. Comput. Vis. Pattern Recognit. (CVPR)},
  pages={7876--7885},
  year={2023}
}

@inproceedings{chen2023nagphormer,
  title={{NAGphormer}: A tokenized graph transformer for node classification in large graphs},
  author={Chen, Jinsong and Gao, Kaiyuan and Li, Gaichao and He, Kun},
  booktitle={Proc. Int. Conf. Learn. Represent. (ICLR)},
  year={2023}
}

@inproceedings{ma2024polyformer,
  title={{PolyFormer}: Scalable node-wise filters via polynomial graph transformer},
  author={Ma, Jiahong and He, Mingguo and Wei, Zhewei},
  booktitle={Proc. ACM SIGKDD Int. Conf. Knowl. Discovery Data Mining (KDD)},
  pages={2118--2129},
  year={2024}
}

@inproceedings{morris2020tudataset,
  title={{TUDataset}: A collection of benchmark datasets for learning with graphs},
  author={Morris, Christopher and Kriege, Nils M. and Bause, Franka and Kersting, Kristian and Mutzel, Petra and Neumann, Marion},
  booktitle={ICML 2020 Workshop on Graph Representation Learning and Beyond},
  year={2020}
}

@inproceedings{yang2016revisiting,
  title={Revisiting semi-supervised learning with graph embeddings},
  author={Yang, Zhilin and Cohen, William W. and Salakhutdinov, Ruslan},
  booktitle={Proc. Int. Conf. Mach. Learn. (ICML)},
  pages={40--48},
  year={2016}
}

@article{shchur2018pitfalls,
  title={Pitfalls of graph neural network evaluation},
  author={Shchur, Oleksandr and Mumme, Maximilian and Bojchevski, Aleksandar and G{\"u}nnemann, Stephan},
  journal={arXiv preprint arXiv:1811.05868},
  year={2018}
}

@inproceedings{mernyei2020wikics,
  title={{Wiki-CS}: A {W}ikipedia-based benchmark for graph neural networks},
  author={Mernyei, P{\'e}ter and Cangea, C{\u{a}}t{\u{a}}lina},
  booktitle={ICML 2020 Workshop on Graph Representation Learning and Beyond},
  year={2020}
}

@inproceedings{pei2020geom,
  title={{Geom-GCN}: Geometric graph convolutional networks},
  author={Pei, Hongbin and Wei, Bingzhe and Chang, Kevin Chen-Chuan and Lei, Yu and Yang, Bo},
  booktitle={Proc. Int. Conf. Learn. Represent. (ICLR)},
  year={2020}
}

@book{penrose2003random,
  title={Random Geometric Graphs},
  author={Penrose, Mathew},
  volume={5},
  year={2003},
  publisher={OUP Oxford}
}

@article{thanou2017learning,
  title={Learning heat diffusion graphs},
  author={Thanou, Dorina and Dong, Xiaowen and Kressner, Daniel and Frossard, Pascal},
  journal={IEEE Trans. Signal Inf. Process. Netw.},
  volume={3},
  number={3},
  pages={484--499},
  year={Sep. 2017},
  publisher={IEEE}
}

@article{lecun1998gradient,
  title={Gradient-based learning applied to document recognition},
  author={LeCun, Yann and Bottou, L{\'e}on and Bengio, Yoshua and Haffner, Patrick},
  journal={Proc. IEEE},
  volume={86},
  number={11},
  pages={2278--2324},
  year={Nov. 1998}
}

\end{document}